\documentclass[leqno, twocolumn, twoside, english]{article}
\usepackage[algo2e, ruled, vlined, linesnumbered]{algorithm2e}
\usepackage{amsmath}
\usepackage{babel}
\usepackage{blindtext}
\usepackage{booktabs}
\usepackage{enumerate}
\usepackage[letterpaper]{geometry}
\usepackage{KITcolors}
\usepackage{ltexpprt}
\usepackage{savefnmark}
\usepackage[group-four-digits, per-mode=symbol, table-parse-only, table-alignment=right]{siunitx}
\usepackage{tikz}

\usepackage[breaklinks, pdfborder={0 0 0}]{hyperref}
\usepackage[capitalise, noabbrev]{cleveref}

\DontPrintSemicolon
\usetikzlibrary{intersections, backgrounds, calc, fit, shapes.geometric}

\crefname{@theorem}{Theorem}{Theorems}
\Crefname{@theorem}{Theorem}{Theorems}

\newcommand*\SetMark[1]{\tikz[remember picture, overlay]{\node (#1) {};}}
\newcommand*\DrawBox[2]{\tikz[remember picture, overlay]{\node [
  fill=#2, inner sep=0pt, opacity=.15, fit={
    ($(_NW) + (-1.5em,.7\baselineskip + 2pt)$)
    ($(_SW) + (-1.5em + \linewidth,.7\baselineskip + #1)$)
  }
] {};}}

\newcommand*\BeginBox{\vspace*{2pt}\SetMark{_NW}}
\newcommand*\EndBoxBeginBox[1]{\vspace*{6pt}\SetMark{_SW}\DrawBox{4pt}{#1}\SetMark{_NW}}
\newcommand*\EndBox[1]{\vspace*{4pt}\SetMark{_SW}\DrawBox{2pt}{#1}}

\SetKwComment{BoxComment}{}{}
\newcommand*\BoxHeading[3]{\BoxComment*[#1]{\textbf{\textrm{\textcolor{#2}{#3}}}\hspace{-1em}}}

\title{\Large Fast, Exact and Scalable Dynamic Ridesharing\thanks{This work was funded by Robert Bosch GmbH, Corporate Sector Research and Advance Engineering.}}
\author{
  Valentin Buchhold\thanks{Karlsruhe Institute of Technology.}\saveFN{\affiliation} \and
  Peter Sanders\useFN{\affiliation} \and
  Dorothea Wagner\useFN{\affiliation}
}
\date{}

\newcommand*\Dist{\mathit{dist}}
\newcommand*\Rank{\mathit{rank}}

\newcommand*\InitialLocation{l_\mathsf{i}}
\newcommand*\CurrentLocation{l_\mathsf{c}}
\newcommand*\StopTime{t_\mathsf{stop}}
\newcommand*\Now{t_\mathsf{now}}

\newcommand*\DepTime{t_\mathsf{dep}}
\newcommand*\ArrTime{t_\mathsf{arr}}
\newcommand*\MinDepTime[1][]{\DepTime^\mathsf{min#1}}
\newcommand*\MaxDepTime[1][]{\DepTime^\mathsf{max#1}}

\newcommand*\MaxArrTime[1][]{\ArrTime^\mathsf{max#1}}
\newcommand*\MinServTime{t_\mathsf{serv}^\mathsf{min}}
\newcommand*\MaxServTime{t_\mathsf{serv}^\mathsf{max}}
\newcommand*\MaxWaitTime{t_\mathsf{wait}^\mathsf{max}}
\newcommand*\MaxTripTime{t_\mathsf{trip}^\mathsf{max}}

\newcommand*\WaitViolationWeight{\gamma_\mathsf{wait}}
\newcommand*\TripViolationWeight{\gamma_\mathsf{trip}}

\newcommand*\Detour{\delta}
\newcommand*\PickupDetour{\Detour_\mathsf{p}}
\newcommand*\DropoffDetour{\Detour_\mathsf{d}}

\newcommand*\Parent{\mathit{parent}}
\newcommand*\SourceBucket{B_\mathsf{s}}
\newcommand*\TargetBucket{B_\mathsf{t}}

\newcommand*\VehicleSpd{v_\mathsf{veh}}
\newcommand*\DistScalingParam{\sigma_\mathsf{dist}}
\newcommand*\SpdScalingParam{\sigma_\mathsf{spd}}

\begin{document}

\maketitle

\begin{abstract}
  \small\baselineskip=9pt
  We study the problem of servicing a set of ride requests by dispatching a set of shared vehicles, which is faced by ridesharing companies such as Uber and Lyft. Solving this problem at a large scale might be crucial in the future for effectively using large fleets of autonomous vehicles. Since finding a solution for the entire set of requests that minimizes the total driving time is NP-complete, most practical approaches process the requests one by one. Each request is inserted into any vehicle's route such that the increase in driving time is minimized. Although this variant is solvable in polynomial time, it still takes considerable time in current implementations, even when inexact filtering heuristics are used. In this work, we present a novel algorithm for finding best insertions, based on (customizable) contraction hierarchies with local buckets. Our algorithm finds provably exact solutions, is still 30~times faster than a state-of-the-art algorithm currently used in industry and academia, and scales much better. When used within iterative transport simulations, our algorithm decreases the simulation time for largescale scenarios with many requests from days to hours.
\end{abstract}

\section{Introduction.}
\label{sec:introduction}

Taxi-like transport options such as cabs, minibuses, rickshaws and ridesharing services already play a vital role in meeting the transport demand in metropolitan areas. They may become even more important in the presence of intelligent ridesharing software, autonomous vehicles, and the desire to combat traffic jams, accidents, air pollution, and lack of sufficient parking. With many thousands and eventually millions of vehicles and riders, this yields fairly complex combinatorial optimization problems that have to be solved in real time. In order to evaluate the impact of ridesharing on people, the environment and the economy, we also have to simulate large realistic scenarios \emph{now}. This requires processing millions of ride requests again and again. For example, one of the leading transport simulators~\cite{HorniNA16} performs hundreds of runs in order to compute realistic activity-travel patterns that describe how travelers behave under certain assumptions.

Current approaches to solve the ridesharing problem require a huge number of calls to Dijkstra's shortest-path algorithm. These are prohibitively expensive for large-scale transport simulations and they are a limiting factor for real-time dispatching of large fleets in metropolitan areas. The goal of this work is to show how to replace Dijkstra's classic algorithm with much faster route planning algorithms.

Ridesharing problems come in a wide variety with different assumptions, objectives, and constraints. To make our work tractable and concrete, we focus on one particular scenario adopted by a leading group in transport simulation~\cite{BischoffMN17, HorniNA16}. This scenario mimics a ridesharing service that answers real-time requests for immediate rides from a given source to a given target. The dispatching algorithm knows the current routes of a fleet of vehicles, each of which has a certain number of seats. The algorithm tries all possible ways to insert a ride request into each vehicle's route. The objective is to minimize the total operation time of the fleet. There are also constraints on the maximum wait time and the maximum time when a rider should reach their target. The best insertion that satisfies all constraints is selected. We use a network with scalar (time-independent) travel times. However, by building on customizable contraction hierarchies~\cite{DibbeltSW16}, we can quickly update these costs according to the current traffic situation every few minutes.

Our novel dispatching algorithm LOUD (for \underline{lo}cal b\underline{u}ckets \underline{d}ispatching) adapts bucket-based contraction hierarchies~\cite{KnoppSSSW07} developed for many-to-many shortest-path computations to the ridesharing problem. 

Contraction hierarchies (CHs)~\cite{GeisbergerSSV12} are a point-to-point route planning technique that is much faster than Dijkstra's algorithm (four orders of magnitude on continental networks). CHs replace systematic exploration of \emph{all} vertices in the network with two much smaller search spaces (forward and reverse) in directed acyclic graphs, in which each edge leads to a ``more important'' vertex. Customizable contraction hierarchies (CCHs)~\cite{DibbeltSW16} are a variant of CHs that can handle updates to the edge costs quickly (e.g., to support real-time traffic updates).

CHs with buckets (BCHs)~\cite{KnoppSSSW07} extend standard and customizable CHs to the many-to-many shortest-path problem by storing CH search spaces in buckets. More precisely, if $v$ appears in a search space from $s$ with distance~$x$, then $(s, x)$ is stored in a \emph{bucket}~$B(v)$ associated with $v$. For example, assume that we have stored the forward search spaces of a set~$S$ of vertices in buckets. Now, we can perform a many-to-one query (from $S$ to a vertex~$t$) by computing the reverse CH search space from $t$. For each vertex~$v$ in the search space with distance~$y$ to $t$, we scan the bucket~$B(v)$. For each entry~$(s, x) \in B(v)$, we obtain $x + y$ as a candidate for the shortest-path distance from $s$ to $t$.

Geisberger et al.~\cite{GeisbergerLSNV10} adapt BCH to a simple carpooling problem, where drivers with a fixed source and target can pick up and drop off passengers heading the same way, as a means of sharing the costs of travel. Their problem, however, is very simplistic. The authors neglect departure times, vehicles shared with more than one passenger, and vehicles already on their way.

\subsubsection*{Our Contribution.}

We present LOUD, a novel algorithm for the problem outlined above. LOUD maintains the forward and reverse CH search spaces of all scheduled (but not completed) pickups and dropoffs in buckets. From these buckets, LOUD can quickly obtain the cost of each possible insertion (i.e., the increase in operation time that is caused by the insertion).

One of our main contributions is a technique to aggressively prune the buckets, so that only those entries remain that can possibly contribute to feasible insertions. This technique decreases the search-space size by a factor of more than 20. Another major contribution is a filtering technique that restricts the search for the best insertion to a small set of promising vehicles. We stress that both techniques do not sacrifice optimality. A contribution that is also applicable to other dispatching algorithms is a data structure for checking whether an insertion into a vehicle's route satisfies the constraints of each rider assigned to the same vehicle. We can do this in constant time, independent of the number of riders assigned to the vehicle.

We extensively evaluate LOUD on the state-of-the-art Open Berlin Scenario~\cite{ZiemkeKN19} and a second, even larger benchmark instance. The experimental results show that LOUD is 30~times faster than algorithms currently used in industry and academia. When used in a transport simulator that performs hundreds of runs, the simulation time decreases from days to hours.

\subsubsection*{Related Work.}

Dynamic ridesharing is related to the classic \emph{dial-a-ride problem} (DARP) in operations research; see \cite{CordeauL07, HoSKLPT18} for recent overviews. The DARP literature, however, primarily considers the \emph{static} variant (where all ride requests are known in advance), often defines the problem on a complete graph, and frequently solves only small instances (using integer linear programming methods in many cases). For these reasons, most DARP approaches are unsuitable for modern largescale ridesharing services.

Finding a solution for an entire set of ride requests that minimizes the total driving time is NP-complete by reduction from the traveling salesman problem with time windows~\cite{MaZW13, Savelsbergh85}. Jung et al.~\cite{JungJP16} propose a simulated-annealing algorithm for this problem. More scalable approaches insert the requests one by one into any vehicle's route while leaving all other vehicle routes unchanged (often using inexact filtering heuristics).

The dispatching algorithm~\cite{BischoffMN17} used by the transport simulation \emph{MATSim}~\cite{HorniNA16} works in three phases. Given a ride request, the first phase tries \emph{all} possible insertions into \emph{each} vehicle's route. For efficiency, all needed detour times are estimated using geometric distances. The second phase uses Dijkstra's algorithm~\cite{Dijkstra59} to compute exact detour times for each insertion that is feasible based on the detour estimates. The last phase evaluates all filtered insertions again (now using exact detour times) and picks the best insertion among those.

The \emph{T-Share algorithm}~\cite{MaZW13} partitions the network into cells using a grid and precomputes the shortest-path distance between all cell centers. To quickly find a heuristic set of candidate vehicles, T-Share searches cells close to the request's source and target cell. For each candidate vehicle, T-Share tries all possible insertions. Each insertion is first evaluated using detour estimates based on precomputed distances, and if it looks feasible, T-Share computes exact (shortest-path) detour times.

Huang et al.~\cite{HuangBJW14} also use grid partitions to find a heuristic set of candidate vehicles. They allow to reorder requests already assigned to a vehicle. Shortest-path distances are computed using a very fast point-to-point routing algorithm (hub labeling~\cite{AbrahamDGW11}) and caching.

A special case of dynamic ridesharing is \emph{dynamic carpooling}, a problem faced by carpooling services such as BlaBlaCar. In this case, the vehicle routes are not determined solely by the passengers. Instead, each driver has a fixed source and target and can pick up and drop off passengers heading the same way, as a means of sharing the costs of travel. Moreover, all constraints (such as an upper bound on the detour time) apply not only to passengers but also to drivers.

Pelzer et al.~\cite{PelzerXZLKA15} partition the network along main roads into cells. For each vehicle, they maintain the sequence of cells through which the vehicle will pass (its \emph{corridor}). A vehicle is a candidate for servicing a given ride request if the pickup is in the same cell as the vehicle and the dropoff is in the corridor of the vehicle. For each candidate vehicle, the authors compute exact detour times using Dijkstra's shortest-path algorithm.

The carpooling algorithm by Geisberger et al.~\cite{GeisbergerLSNV10} is based on the route planning technique contraction hierarchies (CHs)~\cite{GeisbergerSSV12}. It stores the forward and reverse CH search space of each vehicle's source and target, respectively, in buckets~\cite{KnoppSSSW07}. Given a ride request, the buckets are used to compute exact detour times for \emph{all} vehicles. The studied problem, however, is very simplistic. The authors neglect departure times and can match neither more than one request with the same vehicle nor vehicles that are already on their way. Abraham et al.~\cite{AbrahamDFGW12} solve the same simplistic problem in a database, with CH search spaces stored in tables.

Herbawi and Weber~\cite{HerbawiW12} combine an insertion-based algorithm with periodic reoptimizations using a relatively slow evolutionary algorithm.

There has also been previous work on \emph{multi-hop} carpooling~\cite{DrewsL13, MasoudJ17}, where passengers can transfer from one vehicle to another as part of a single journey. These algorithms model the problem as a time-expanded graph~\cite{PallottinoS98}, similar to graph-based techniques for journey planning in public transit networks~\cite{SchulzWW00, BastCEGHRV10, DellingDPW15}. To avoid combinatorial explosion, however, they need to discretize both space and time. That is, they do not support door-to-door transport and departures, arrivals and transfers can only happen at interval endpoints. Despite these limitations, the matching algorithms are relatively slow, even on medium-sized instances.

\subsubsection*{Outline.}

This work is organized as follows. \Cref{sec:problem-statement} provides a precise definition of the basic problem we solve. \Cref{sec:preliminaries} briefly reviews crucial building blocks LOUD builds on. \Cref{sec:approach} describes LOUD in detail, including extensions to meet additional requirements of real-world production systems. \Cref{sec:experiments} presents an extensive experimental evaluation on various benchmark instances, which includes a comparison to related work. \Cref{sec:conclusion} concludes with final remarks.

\section{Problem Statement.}
\label{sec:problem-statement}

This section defines the basic problem we consider. Potential extensions will be discussed in \cref{sec:extensions}.

We treat a road network as a directed graph~$G = (V, E)$ where vertices represent intersections and edges represent road segments. Each edge~$(v, w) \in E$ has a nonnegative length~$\ell(v, w)$ representing the travel time between $v$ and $w$. Note that we denote by $\Dist(v, w)$ the shortest-path distance (i.e., travel time) from $v$ to $w$.

We are given a set of vehicles. Each vehicle~$\nu = (\InitialLocation, c, \MinServTime, \MaxServTime)$ has an initial location~$\InitialLocation$, a seating capacity~$c$, and a service interval~$[\MinServTime, \MaxServTime)$. For each vehicle~$\nu$, we maintain its route~$R(\nu) = \langle s_0, \dots, s_k \rangle$, which is a sequence of stops~$s$ at locations~$l(s) \in V$ that are already scheduled for the vehicle. At each stop, the vehicle picks up and/or drops off one or more riders. Independent of the number of riders boarding and alighting, each stop takes time~$\StopTime$. Each vehicle's route is continuously updated according to the current situation. More precisely, if a vehicle~$\nu$ is currently making a stop, then $s_0$ is the current stop. If a vehicle~$\nu$ is currently driving, then $s_0$ is the previous stop (i.e., the vehicle's current location~$\CurrentLocation(\nu)$ is somewhere between $s_0$ and $s_1$). Idle vehicles prolong their last stop. Abusing notation, we sometimes use stops as vertices. For example, $\Dist(s, s')$ is a shorthand for $\Dist(l(s), l(s'))$.

We consider a scenario in which a dispatching server receives ride requests and immediately matches them to vehicles. Each request~$r = (p, d, \MinDepTime)$ has a pickup spot~$p \in V$, a dropoff spot~$d \in V$, and an earliest departure time~$\MinDepTime$. We do not allow pre-booking, i.e., each ride request is submitted, received and matched at $\MinDepTime$. Note that this is by far the most common scenario, adopted by the leading ridehailing services Uber and Lyft and also by related work~\cite{BischoffMN17, MaZW13, HuangBJW14, JungJP16}. The goal is to insert each request into any vehicle's route such that the vehicle's detour~$\Detour$ (i.e., the increase in operation time) is minimized. Formally, an insertion can be described by a quadruple~$(\nu, r, i, j)$ indicating that vehicle~$\nu$ picks up request~$r$ immediately after stop~$s_i(\nu)$ and drops off $r$ immediately after stop~$s_j(\nu)$. Besides capacity and service time constraints, the insertion is subject to two additional constraints.

\begin{enumerate}[(1)]
  \item The \emph{wait time} for each request~$r'$ already matched to the vehicle must not exceed a certain threshold, i.e., after the insertion the vehicle must still pick up request~$r'$ no later than $\MaxDepTime(r') = \MinDepTime(r') + \MaxWaitTime$, where $\MaxWaitTime$ is a model parameter.
  \item The \emph{trip time} for each request~$r'$ already matched to the vehicle must not exceed a certain threshold, i.e., after the insertion the vehicle must still drop off $r'$ no later than $\MaxArrTime(r') = \MinDepTime(r') + \MaxTripTime(r') = \MinDepTime(r') + \alpha \cdot \Dist(p(r'), d(r')) + \beta$, where $\alpha$ and $\beta$ are model parameters as well.
\end{enumerate}

For each request already matched to the vehicle, (1) and (2) are \emph{hard} constraints, i.e., they must always be satisfied. If any wait or trip time constraint is violated, the insertion is feasible only if it leads to no additional delay for any already matched request. For the request~$r$ to be inserted, (1) and (2) are \emph{soft} constraints, i.e., they may be violated. However, the violation of the wait time constraint and the violation of the trip time constraint are added to the objective value. More precisely, the objective value~$f(\iota)$ of an insertion~$\iota$ is
\begin{equation}
  \label{eq:objective-value}
  \begin{split}
    f(\iota) = \Detour &+ \WaitViolationWeight \cdot \max\{\DepTime(p(r)) - \MaxDepTime(r), 0\}\\
                       &+ \TripViolationWeight \cdot \max\{\ArrTime(d(r)) - \MaxArrTime(r), 0\},
  \end{split}
\end{equation}
where $\DepTime(p(r))$ is the scheduled departure time at the pickup spot, $\ArrTime(d(r))$ is the scheduled arrival time at the dropoff spot, and $\WaitViolationWeight$ and $\TripViolationWeight$ are parameters.

Whenever a request is received, the goal is to find the insertion $\iota$ into any route that minimizes $f(\iota)$. If there is no feasible insertion, the request is rejected. However, since the wait and trip time constraint are soft for the request to be inserted, a request is rejected only if all vehicles go out of service before the request can be served. With unbounded service intervals (which are feasible for driverless vehicles), no requests are rejected.

\section{Preliminaries.}
\label{sec:preliminaries}

A crucial building block of LOUD are bucket-based contraction hierarchies. In the following, we first briefly review Dijkstra's shortest-path algorithm and then discuss contraction hierarchies and customizable contraction hierarchies, which are both speedup techniques for Dijkstra. Finally, we consider bucket-based (customizable) contraction hierarchies, an extension to batched shortest paths such as the one-to-many and many-to-many shortest-path problem.

\subsection{\hspace{-1pt}Dijkstra's Algorithm.}

\emph{Dijkstra's algorithm}~\cite{Dijkstra59} computes the shortest-path distances from a source vertex~$s$ to all other vertices. For each vertex~$v$, it maintains a \emph{distance label}~$d_s(v)$, which represents the length of the shortest path from $s$ to $v$ seen so far. Moreover, it maintains an addressable priority queue~$Q$~\cite{SandersMDD19} of vertices, using their distance labels as keys. Initially, $d_s(s) = 0$ for the source~$s$, $d_s(v) = \infty$ for each vertex~$v \ne s$, and $Q = \{s\}$.

The algorithm repeatedly extracts a vertex~$v$ with minimum distance label from the queue and \emph{settles} it by \emph{relaxing} its outgoing edges~$(v, w)$. To relax an edge~$e = (v, w)$, the path from $s$ to $w$ via $v$ is compared with the shortest path from $s$ to $w$ found so far. More precisely, if $d_s(v) + \ell(e) < d_s(w)$, the algorithm sets $d_s(w) = d_s(v) + \ell(e)$ and inserts $w$ into the queue. It stops when the queue becomes empty.

\subsection{Contraction Hierarchies.}

\emph{Contraction hierarchies} (CHs)~\cite{GeisbergerSSV12} are a two-phase speedup technique to accelerate point-to-point shortest-path computations, which exploits the inherent hierarchy of road networks. To differentiate them from customizable CHs, we sometimes call them \emph{weighted} or \emph{standard} CHs. The preprocessing phase heuristically orders the vertices by importance, and \emph{contracts} them from least to most important. Intuitively, vertices that hit many shortest paths are considered more important, such as vertices on highways. To contract a vertex~$v$, it is temporarily removed from the graph, and \emph{shortcut} edges are added between its neighbors to preserve distances in the remaining graph (without $v$). Note that a shortcut is only needed if it represents the only shortest path between its endpoints, which can be checked by running a \emph{witness search} (local Dijkstra) between its endpoints.

The query phase performs a bidirectional Dijkstra search on the augmented graph that only relaxes edges leading to vertices of higher \emph{ranks} (importance). More precisely, let a \emph{forward CH search} be a Dijkstra search that relaxes only outgoing upward edges, and a \emph{reverse CH search} one that relaxes only incoming downward edges. A \emph{CH query} runs a forward CH search from the source and a reverse CH search from the target until the search frontiers meet. The stall-on-demand~\cite{GeisbergerSSV12} optimization prunes the search at any vertex~$v$ with a suboptimal distance label, which can be checked by looking at the downward edges coming into $v$.

\subsection{\hspace{-4pt}Customizable Contraction Hierarchies.}

\emph{Customizable contraction hierarchies} (CCHs)~\cite{DibbeltSW16} are a three-phase technique, splitting CH preprocessing into a relatively slow metric-independent phase and a much faster customization phase. The metric-independent phase computes a \emph{separator decomposition}~\cite{BauerCRW16} of the unweighted graph, determines an associated \emph{nested dissection order}~\cite{George73} on the vertices, and contracts them in this order without running witness searches (which depend on the metric). Therefore, it adds every potential shortcut. The customization phase computes the lengths of the edges in the hierarchy by processing them in bottom-up fashion. To process an edge~$(u, w)$, it enumerates all triangles~$\{v, u, w\}$ where $v$ has lower rank than $u$ and $w$, and checks whether the path~$\langle u, v, w \rangle$ improves the length of $(u, w)$. Alternatively, Buchhold et al.~\cite{BuchholdSW19} enumerate all triangles~$\{u, w, v'\}$ where $v'$ has higher rank than $u$ and $w$, and check if the path~$\langle v', u, w \rangle$ improves the length of $(v', w)$, accelerating the customization phase by a factor of 2.

There are two query algorithms. First, one can run a standard CH query. Second, there is a query algorithm based on the \emph{elimination tree} of the hierarchy. The parent of a vertex in the elimination tree is the lowest-ranked of its higher-ranked neighbors in the hierarchy. Bauer et al.~\cite{BauerCRW16} prove that the ancestors of a vertex~$v$ in the elimination tree are exactly the set of vertices scanned by a Dijkstra-based CCH search from $v$. An elimination tree search from $v$ therefore scans all vertices in the CCH search space of $v$ in order of increasing rank by traversing the path in the elimination tree from $v$ to the root. As elimination tree queries use no priority queues, they are usually faster than Dijkstra-based CCH queries. Buchhold et al.~\cite{BuchholdSW19} propose further optimizations for the elimination tree search. 

\subsection{CHs with Buckets.}

The \emph{bucket-based approach} by Knopp et al.~\cite{KnoppSSSW07} extends any hierarchical speedup technique such as CHs and CCHs to batched shortest paths. In the one-to-many shortest-path problem, the goal is to compute shortest paths from a source~$s \in V$ to each target~$t \in T \subseteq V$. A bucket-based CH (BCH) search maintains a tentative distance~$D_s(t)$ from $s$ to each $t$, initialized to $\infty$, and for each vertex~$h$ an initially empty bucket~$B(h)$. First, the algorithm runs a reverse CH search from each $t$ and inserts, for each vertex $h$ settled, an entry~$(t, d_t(h))$ into $B(h)$. Note that $(t, d_t(h))$ can be thought of as a shortcut from $h$ to $t$ with length~$d_t(h)$. Then, the algorithm runs a forward CH search from $s$ and loops, for each vertex~$h$ settled, over all entries~$(t, d_t(h)) \in B(h)$. If $d_s(h) + d_t(h) < D_s(t)$, it sets $D_s(t) = d_s(h) + d_t(h)$. Many-to-one queries from each source~$s \in S \subseteq V$ to a target~$t \in V$ work analogously. In this case, each bucket~$B(h)$ stores shortcuts from several $s$ to $h$.

\section{Our Approach.}
\label{sec:approach}

We begin with a high-level description of LOUD, our new algorithm for dispatching a fleet of shared vehicles. Let $r = (p, d, \MinDepTime)$ be the ride request to be inserted and let $\nu$ be a vehicle with route~$R(\nu) = \langle s_0, \dots, s_k \rangle$. We will ignore some special cases for now but will discuss them later. In particular, we defer insertions~$(\nu, r, i, j)$ with $i = 0$ or $j = k$ to \cref{sec:special-cases}.

To find the best insertion for request~$r$, we consider a superset~$C$ of the vehicles~$\nu$ that allow at least one feasible insertion~$(\nu, r, i, j)$ with $i \ne k$. For each vehicle~$\nu \in C$, we look at all insertions~$(\nu, r, i, j)$ with $0 < i \le j < k$. For each such insertion, we check whether the hard constraints are satisfied and compute the insertion cost according to \cref{eq:objective-value}, i.e., the vehicle's detour plus the violations of the soft constraints (if any). When the algorithm stops, we return the best feasible insertion seen so far.

To compute the cost of an insertion~$(\nu, r, i, j)$, we generally need the distance~$\Dist(s_i, p)$ from stop~$s_i$ to the pickup spot~$p$, the distance~$\Dist(p, s_{i + 1})$ from $p$ to stop~$s_{i + 1}$, the distance~$\Dist(s_j, d)$ from stop~$s_j$ to the dropoff spot~$d$, and finally the distance~$\Dist(d, s_{j + 1})$ from $d$ to stop~$s_{j + 1}$. We propose using BCHs to compute these distances. For each vertex~$h$, we maintain a \emph{source bucket}~$\SourceBucket(h)$ and a \emph{target bucket}~$\TargetBucket(h)$, both initially empty. Whenever we insert a stop~$s$ into a vehicle's route, we run a forward (reverse) CH search from $s$ and insert, for each vertex~$h$ settled by the search, an entry~$(s, d_s(h))$ into $\SourceBucket(h)$ ($\TargetBucket(h)$). When we receive request~$r$, we run two forward BCH searches (from $p$ and from $d$) that scan the target buckets, and two reverse BCH searches (from $p$ and from $d$) that scan the source buckets. This gives us the distances we need to compute the costs of all candidate insertions.

We are now ready to introduce one of the main ideas of LOUD. We observe that the leeway~$\lambda$ between each pair of consecutive stops we have to insert new stops is bounded, due to the hard constraints for the requests already matched to a vehicle. That is, we are not allowed to take arbitrarily long detours between two consecutive stops on a vehicle's route. See \cref{fig:elliptic-pruning} for an illustration. Each additional stop~$s$ we may insert between stops~$s_i$ and $s_{i + 1}$ has to lie inside a \emph{shortest-path ellipse}, defined as the set of vertices~$v$ with $\Dist(s_i, v) + \Dist(v, s_{i + 1}) \le \lambda$ (i.e., $s_i$ and $s_{i + 1}$ are the foci of the ellipse). Naturally, the entire shortest path from $s_i$ via $s$ to $s_{i + 1}$ has to lie inside the ellipse. Hence, when computing source bucket entries from $s_i$, we need to insert an entry~$(s_i, d_{s_i}(h))$ into $\SourceBucket(h)$ only if $h$ lies inside the ellipse around $s_i$ and $s_{i + 1}$. Target bucket entries can be pruned analogously. We call this \emph{elliptic pruning} and it is surprisingly effective, as our experiments in \cref{sec:experiments} will show.

\begin{figure}[tb]
  \centering
  \begin{tikzpicture}
  [vertex/.style={draw, fill, inner sep=0pt},
   stop/.style={vertex, circle, minimum size=1ex},
   source hub/.style={vertex, rectangle, minimum size=1 / sqrt(2) * 1ex},
   target hub/.style={vertex, diamond, minimum size=1ex},
   path/.style={>=stealth, shorten >=1pt, shorten <=1pt},
   possible path/.style={path, dash pattern=on 1.5pt off 1.5pt},
   leg/.style={path},
   leeway/.style 2 args={#1, draw, fill, fill opacity=#2},
   scale=.65]
  \def\Alpha{1.5}

  \begin{scope}[rotate=-30]
    \node [stop] (s1) at (-2.5,0) {};
    \node [stop] (s2) at ( 2.5,0) {} edge [leg, <-] (s1);

    \path [KITlilac] (s1) +( 48:2.0) node [source hub] {} edge [path, <-] (s1);
    \path [KITlilac] (s1) +( 96:1.5) node [source hub] {} edge [path, <-] (s1);
    \path [KITpalegreen] (s1) +(144:0.5) node [target hub] {} edge [path, ->] (s1);
    \path [KITpalegreen] (s1) +(192:1.5) node [target hub] (p-s1 hub) {} edge [path, ->] (s1);
    \path [KITlilac] (s1) +(280:1.0) node [source hub] {} edge [path, <-] (s1);
    \path [KITpalegreen] (s1) +(320:1.5) node [target hub] {} edge [path, ->] (s1);

    \path [KITlilac] (s2) +(100:2.0) node [target hub] {} edge [path, ->] (s2);
    \path [KITcyanblue] (s2) +(140:1.0) node [source hub] {} edge [path, <-] (s2);
    \path [KITlilac] (s2) +(228:2.0) node [target hub] {} edge [path, ->] (s2);
    \path [KITlilac] (s2) +(276:1.5) node [target hub] {} edge [path, ->] (s2);
    \path [KITcyanblue] (s2) +(324:0.5) node [source hub] {} edge [path, <-] (s2);
    \path [KITcyanblue] (s2) +(372:1.0) node [source hub] (s2-d hub) {} edge [path, <-] (s2);

    \begin{scope}[on background layer]
      \path [leeway={KITlilac}{0.10}]
          ellipse [x radius=\Alpha / 2 * 5, y radius=sqrt(\Alpha^2 - 1) / 2 * 5];
    \end{scope}
  \end{scope}

  \begin{scope}[shift={(-150:2)}, shift={(150:2.5)}, rotate=30]
    \node [stop] (s0) at (-2,0) {} edge [leg, ->] (s1);

    \path [KITpalegreen] (s0) +( 72:1.0) node [source hub] (s0-p hub) {} edge [path, <-] (s0);
    \path [KITpalegreen] (s0) +(144:0.5) node [source hub] {} edge [path, <-] (s0);
    \path [KITpalegreen] (s0) +(288:1.5) node [source hub] {} edge [path, <-] (s0);

    \node [stop, label=$p$] (p) at (0,{sqrt(\Alpha^2 - 1) / 4 * 4}) {}
        edge [possible path, <-] (s0-p hub)
        edge [possible path, ->] (p-s1 hub);

    \begin{scope}[on background layer]
      \path [leeway={KITpalegreen}{0.15}]
          ellipse [x radius=\Alpha / 2 * 4, y radius=sqrt(\Alpha^2 - 1) / 2 * 4];
    \end{scope}
  \end{scope}

  \begin{scope}[shift={(30:1.5)}, shift={(-30:2.5)}, rotate=30]
    \node [stop] (s3) at (1.5,0) {} edge [leg, <-] (s2);

    \path [KITcyanblue] (s3) +(252:1.0) node [target hub] (d-s3 hub) {} edge [path, ->] (s3);
    \path [KITcyanblue] (s3) +(396:0.5) node [target hub] {} edge [path, ->] (s3);
    \path [KITcyanblue] (s3) +(468:1.0) node [target hub] {} edge [path, ->] (s3);

    \node [stop, label=below:$d$] (p) at (0,-{sqrt(\Alpha^2 - 1) / 4 * 3}) {}
        edge [possible path, <-] (s2-d hub)
        edge [possible path, ->] (d-s3 hub);

    \begin{scope}[on background layer]
      \path [leeway={KITcyanblue}{0.15}]
          ellipse [x radius=\Alpha / 2 * 3, y radius=sqrt(\Alpha^2 - 1) / 2 * 3];
    \end{scope}
  \end{scope}
\end{tikzpicture}%
  \caption{A vehicle's route consisting of four stops and the bucket entries induced by them. The stops are shown as circles and the leeway between two consecutive stops is shown as an ellipse. Source bucket entries are shown as edges with square-shaped heads and target bucket entries are shown as edges with diamond-shaped tails. Green, lilac and blue bucket entries are pruned by the respective ellipse. Consider a request~$r = (p, d, \MinDepTime)$ where $p$ is to be inserted immediately after the first stop~$s_0$ and $d$ immediately before the last stop~$s_3$. Note that the shortest paths from $s_0$ to $s_1$ via $p$ and from $s_2$ to $s_3$ via $d$ lie entirely inside the respective ellipse.}
  \label{fig:elliptic-pruning}
\end{figure}
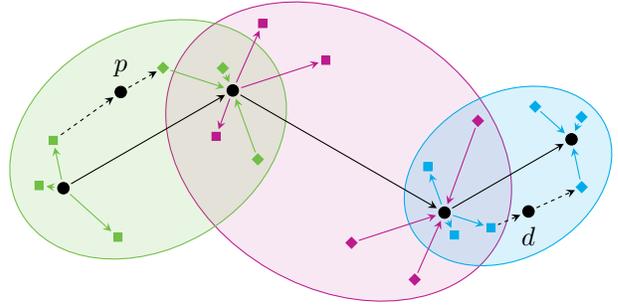

Elliptic pruning has multiple advantages. First, it accelerates the BCH searches, since these searches now scan smaller buckets. Second, it speeds up the removal of bucket entries that refer to completed stops. Note that whenever a vehicle completes a stop, the buckets are updated accordingly. The biggest advantage, however, is that elliptic pruning enables us to obtain a small superset~$C$ of the vehicles~$\nu$ that allow at least one feasible insertion~$(\nu, r, i, j)$ with $i \ne k$. Besides a stop identifier and a distance label, we store in each bucket entry the identifier of the vehicle to which the stop belongs. During the BCH searches, we insert all vehicle identifiers seen into $C$. Without elliptic pruning, the source and target bucket of the highest-ranked vertex in the hierarchy would contain an entry for each stop on each vehicle's route, and thus $C$ would contain each vehicle.

The following sections work out the details of LOUD. \Cref{sec:maintaining-feasibility} discusses how to check whether an insertion is feasible (i.e., satisfies the hard constraints) in constant time. \Cref{sec:elliptic-pruning} shows which bucket entries are necessary and sufficient to find the needed distances, and presents an algorithm that can efficiently check this elliptic pruning criterion. \Cref{sec:special-cases} discusses the special case of insertions~$(\nu, r, i, j)$ with $i = 0$ or $j = k$. \Cref{sec:putting-everything-together} assembles the basic LOUD algorithm from the building blocks introduced in the preceding sections. \Cref{sec:extensions} discusses additional requirements of real-world production systems such as incorporating real-time traffic information into the dispatching server and other potential objective functions.

\subsection{Maintaining Feasibility.}
\label{sec:maintaining-feasibility}

Consider a vehicle's route~$\langle s_0, \dots, s_k \rangle$ and a request~$r = (p, d, \MinDepTime)$. We need a subroutine that checks whether the service time constraint and the wait and trip time constraints for each request assigned to the vehicle are still satisfied when inserting pickup~$p$ immediately after $s_i$ and dropoff~$d$ immediately after $s_j$, $i \le j$. Since this operation is frequently used within LOUD (and even more frequently within competitors such as MATSim), it should be as fast as possible. This section shows how to check all constraints in constant time, independent of the number of stops and the number of requests assigned to the vehicle. Note that current approaches such as MATSim and T-Share take time linear in the length of the route.

For each stop~$s \in R$ on each vehicle route~$R$, we maintain the departure time~$\MinDepTime(s)$ at stop~$s$ when no further stops are inserted into the route. Moreover, we maintain the latest arrival time~$\MaxArrTime(s)$ at stop~$s$ so that all following pickups and dropoffs are on time. Whenever we insert a request~$r' = (p', d', \MinDepTime[\prime])$, yielding a route~$\langle s'_0, \dots, s'_{i'} = p', \dots, s'_{j'} = d', \dots, s'_{k'} \rangle$, we loop over all $s'_\ell$, $i' \le \ell \le k'$, in forward order and set
\begin{equation*}
  \MinDepTime(s'_\ell) = \MinDepTime(s'_{\ell - 1}) + \Dist(s'_{\ell - 1}, s'_\ell) + \StopTime.
\end{equation*}
Furthermore, we set $\MaxArrTime(s'_{i'}) = \MaxDepTime(r') - \StopTime$ as well as $\MaxArrTime(s'_{j'}) = \MaxArrTime(r')$. We propagate these constraints to all preceding stops by looping over all $s'_\ell$, $0 < \ell \le j'$, in reverse order and setting
\begin{equation*}
  \MaxArrTime(s'_\ell) = \min\{\MaxArrTime(s'_\ell), \MaxArrTime(s'_{\ell + 1})
      \scalebox{.5}[1]{$-$} \Dist(s'_\ell, s'_{\ell + 1}) \scalebox{.5}[1]{$-$} \StopTime\}.
\end{equation*}

The $\MinDepTime$ and $\MaxArrTime$ values allow us to check all service, wait and trip time constraints on a route in constant time. We are given a vehicle~$\nu$ with route~$\langle s_0, \dots, s_k \rangle$, a request~$(p, d, \MinDepTime)$, where $p$ is to be inserted immediately after $s_i$ and $d$ is to be inserted immediately after $s_j$, and the distances~$\Dist(s_i, p)$, $\Dist(p, s_{i + 1})$, $\Dist(s_j, d)$, and $\Dist(d, s_{j + 1})$. We first compute the pickup detour time~$\PickupDetour = \Dist(s_i, p) + \StopTime + \Dist(p, s_{i + 1}) - \Dist(s_i, s_{i + 1})$ and the dropoff detour time~$\DropoffDetour = \Dist(s_j, d) + \StopTime + \Dist(d, s_{j + 1}) - \Dist(s_j, s_{j + 1})$. Note that there is no need to store $\Dist(s_i, s_{i + 1})$ and $\Dist(s_j, s_{j + 1})$ explicitly, as they can be obtained from the $\MinDepTime$~values. An insertion then satisfies all time constraints if and only if
\begin{equation*}
  \begin{split}
    \MinDepTime(s_{i + 1}) - \StopTime + \PickupDetour &\le \MaxArrTime(s_{i + 1}) \text{ and}\\
    \MinDepTime(s_{j + 1}) - \StopTime + \PickupDetour + \DropoffDetour
        &\le \MaxArrTime(s_{j + 1}) \text{ and}\\
    \MinDepTime(s_k) + \PickupDetour + \DropoffDetour &\le \MaxServTime(\nu).
  \end{split}
\end{equation*}

An actual implementation needs to treat several special cases. For example, $p$ or $d$ can coincide with an existing stop, $p$ or $d$ can be inserted after $s_k$, or $d$ can be inserted immediately after $p$. However, all these cases are straightforward to implement and we do not discuss them in detail. The correctness of our approach follows directly from \cref{thm:feasibility-correctness}.

\begin{lemma}
  \label{thm:feasibility-correctness}
  All pickups and dropoffs at each stop $s_j$, $j \ge i$, on a vehicle's route are on time if and only if the vehicle arrives at $s_i$ no later than $\MaxArrTime(s_i)$.
\end{lemma}

\begin{proof}
  Let $t$ be the arrival time at $s_i$. We claim that all pickups and dropoffs at each subsequent stop~$s_j$ are on time if $t \le \MaxArrTime(s_i)$. Assume otherwise, that is, there exists a request~$r$ with either $p(r) = s_j$ and $\MaxDepTime(r) < t + \StopTime + \sum_{k = i}^{j - 1}(\Dist(s_k, s_{k + 1}) + \StopTime)$ or $d(r) = s_j$ and $\MaxArrTime(r) < t + \sum_{k = i}^{j - 1}(\Dist(s_k, s_{k + 1}) + \StopTime)$. In the former case, we have
  \begin{equation*}
    \MaxArrTime(s_i) \le \MaxDepTime(r) - \StopTime - \sum_{k = i}^{j - 1}(\Dist(s_k, s_{k + 1}) + \StopTime) < t,
  \end{equation*}
  where the first inequality follows from the construction of $\MaxArrTime(s_i)$ and the second inequality is the assumption. This contradicts $t \le \MaxArrTime(s_i)$. In the latter case, we have
  \begin{equation*}
    \MaxArrTime(s_i) \le \MaxArrTime(r) - \sum_{k = i}^{j - 1}(\Dist(s_k, s_{k + 1}) + \StopTime) < t,
  \end{equation*}
  where the first inequality follows from the construction of $\MaxArrTime(s_i)$ and the second inequality is the assumption. Again, this contradicts that $t \le \MaxArrTime(s_i)$.

  Assume conversely that all pickups and dropoffs at each subsequent stop~$s_j$ are on time. By construction of the $\MaxArrTime$~values, there is a request $r$ with either $\MaxArrTime(s_i) = \MaxDepTime(r) - \StopTime - \sum_{k = i}^{j - 1}(\Dist(s_k, s_{k + 1}) + \StopTime)$ or $\MaxArrTime(s_i) = \MaxArrTime(r) - \sum_{k = i}^{j - 1}(\Dist(s_k, s_{k + 1}) + \StopTime)$. In both cases, we have $\MaxArrTime(s_i) \ge t$ by assumption.
\end{proof}

\subsubsection*{Capacity Constraints.}

Besides service, wait and trip time constraints, we have to handle capacity constraints. To this end, we maintain, for each stop~$s \in R$ on each vehicle route~$R$, the occupancy~$o(s)$ (the number of occupied seats) when the vehicle departs from $s$. Whenever we insert a request~$r' = (p', d', \MinDepTime[\prime])$, yielding a route~$\langle s'_0, \dots, s'_{i'} = p', \dots, s'_{j'} = d', \dots, s'_{k' - 1} \rangle$, we update the occupancies as follows. We first set $o(s'_{i'}) = o(s'_{i' - 1})$ (if $s'_{i'}$ was not present before the insertion of $r'$) and then $o(s'_{j'}) = o(s'_{j' - 1})$ (if $s'_{j'}$ was not present before). Then, we loop over all $s'_\ell$, $i' \le \ell < j'$, and increment $o(s'_\ell)$. We use the $o$~values in \cref{sec:putting-everything-together}.

\subsubsection*{Implementation Details.}

We maintain one dynamic \emph{value array} per stop attribute (such as the stop location~$l$, the earliest departure time~$\MinDepTime$, and the latest arrival time~$\MaxArrTime$), which stores the attribute's value for all stops on all routes. The values for stops on the same route are stored consecutively in memory, in the order in which the stops appear on the route. In addition, all value arrays share a single \emph{index array}, which stores the starting point and ending point of each route's \emph{value block} in the dynamic value arrays.

When we remove a stop from a route, we move the resulting hole in the value arrays to the end of the route's value block, and decrement the block's ending point in the index array. Consider an insertion of a stop into a route. If the element immediately after the route's value block is a hole, we insert the new stop's value into the value block and move the values after the insertion point one position to the right. Analogously, if the element before the value block is a hole, we move the values before the insertion point one position to the left. Otherwise, we move the entire value block to the end of the value arrays, and additionally insert a number of holes after the value block (the number is a constant fraction of the block size). Then, there is a hole after the block, and we proceed as described above.

\subsection{Elliptic Pruning.}
\label{sec:elliptic-pruning}

We use BCHs to obtain the shortest-path distances needed to compute insertion costs, but carefully prune the source and target buckets. Let $s$ and $s'$ be two consecutive stops on a vehicle's route and let $v$ be a new pickup or dropoff spot. The leeway~$\lambda(s, s')$ we have to insert $v$ between $s$ and $s'$ is bounded by $\MaxArrTime(s') - \MinDepTime(s) - \StopTime$. More precisely, inserting $v$ between $s$ and $s'$ is feasible only if $\Dist(s, v) + \Dist(v, s') \le \lambda(s, s')$. Therefore, we only need to find shortest paths from all $s$ to $v$ such that $\Dist(s, v) + \Dist(v, s') \le \lambda(s, s')$. We now show which bucket entries are necessary and sufficient for the reverse BCH search from $v$ to find the needed distances. The case of the forward BCH search from $v$ is symmetric.

\begin{theorem}
  \label{thm:elliptic-pruning}
  Let $s$ and $s'$ be two consecutive stops on a vehicle's route with leeway~$\lambda$ between them. Consider the following two propositions:

  \begin{enumerate}[(1)]
    \item \label{enum:prop-1} For each vertex~$h \in V$, there is an entry~$(s, d_s(h))$ in the source bucket~$\SourceBucket(h)$ if
          \begin{enumerate}[(a)]
            \item $h$ is the highest-ranked vertex on all shortest \mbox{$s$--$h$} paths and
            \item $d_s(h) + \Dist(h, s') \le \lambda$.
          \end{enumerate}
    \item \label{enum:prop-2} A reverse BCH search from $v$ finds a shortest \mbox{$s$--$v$} path for each vertex~$v \in V$ with $\Dist(s, v) + \Dist(v, s') \le \lambda$.
  \end{enumerate}
  Then (\ref{enum:prop-1}) is a necessary and sufficient condition for (\ref{enum:prop-2}).
\end{theorem}

\begin{figure}[tb]
  \centering
  \begin{tikzpicture}
  [vertex/.style={draw, fill, circle, inner sep=0pt, minimum size=1ex},
   path/.style={>=stealth, auto, ->, shorten >=1pt, shorten <=1pt},
   non-existent path/.style={path, dashed},
   path label/.style={anchor=mid},
   P/.style={color=KITgreen},
   Q/.style={color=KITblue},
   R/.style={color=KITcyanblue},
   P'/.style={color=KITpalegreen},
   scale=0.5]

  \path (0,0) node [vertex, label=below:$s$] (s) {}
      ++( 60:5) node [vertex, label=above:$h$ ] (h)  {}
      ++(-60:3) node [vertex, label=below:$v$ ] (v)  {}
      ++( 60:4) coordinate (v-s' hub)
      ++(-60:5) node [vertex, label=below:$s'$] (s') {};

  \path [name path=from s, overlay] (s)  -- +( 75:10);
  \path [name path=to h  , overlay] (h)  -- +(150:10);
  \path [name path=from h, overlay] (h)  -- +( 30:10);
  \path [name path=to s' , overlay] (s') -- +(105:10);

  \path [name intersections={of=from s and to h , by=[{vertex, label=$h'$}]h'}];
  \path [name intersections={of=from h and to s', by=[{vertex, label=$w$ }]w }];

  \draw [path, P, name path=P] (s) -- (h);
  \draw [path, Q, name path=Q] (h) -- (v);
  \draw [path] (v) -- (v-s' hub) -- (s');
  \draw [path, R] (h) -- node {$R$} (w);
  \draw [path] (w) -- (s');
  \draw [non-existent path, P', name path=P'] (s) -- (h');
  \draw [non-existent path, P'] (h') -- (h);

  \path [name path=mid] (60:3) +(-3,0) -- +(3,0);

  \path [name intersections={of=mid and P , by=P }] (P)  +( 1.33em,0) node [path label, P ] {$P$};
  \path [name intersections={of=mid and Q , by=Q }] (Q)  +(-1.33em,0) node [path label, Q ] {$Q$};
  \path [name intersections={of=mid and P', by=P'}] (P') +(-1.33em,0) node [path label, P'] {$P'$};

  \draw [dashed, >=stealth, ->] (-1,0) -- node [above, sloped] {rank} (-1,0 |- w);
\end{tikzpicture}%
  \caption{A possible pickup or dropoff at vertex~$v$ inserted between the consecutive stops $s$ and $s'$.}
  \label{fig:bch-correctness-proof}
\end{figure}
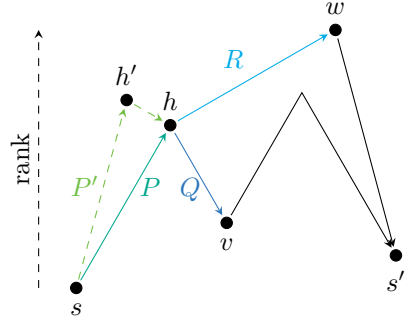

\begin{proof}
  Assume that (\ref{enum:prop-1}) holds and let $v$ be a vertex with $\Dist(s, v) + \Dist(v, s') \le \lambda$ (see \cref{fig:bch-correctness-proof} for an illustration). We say that a path~$P$ is \emph{higher} than a path~$Q$ if $\max_{w \in P} \Rank(w) > \max_{w \in Q} \Rank(w)$. Let $h$ be the highest-ranked vertex on a highest of the shortest $s$--$v$ paths. By construction, there is a shortest $s$--$h$ path~$P$ containing only upward edges and a shortest $h$--$v$ path~$Q$ containing only downward edges, and hence $P \cdot Q$ is an up-down path. We have
  \begin{equation*}
    \begin{split}
      d_s(h) + \Dist(h, s') &= \Dist(s, h) + \Dist(h, s')\\
                            &\le \Dist(s, v) + \Dist(v, s') \le \lambda,
    \end{split}
  \end{equation*}
  where the equality follows from the fact that $P$ contains only upward edges, the first inequality comes from the triangle inequality~$\Dist(h, s') \le \Dist(h, v) + \Dist(v, s')$, and the second inequality uses the definition of $v$. Then $(s, d_s(h)) \in \SourceBucket(h)$ by (\ref{enum:prop-1}), and a reverse BCH search from $v$ finds the shortest $s$--$v$ path~$P \cdot Q$.

  Assume conversely that (\ref{enum:prop-2}) holds and let $h$ be a vertex such that $h$ is the highest-ranked vertex on all shortest $s$--$h$ paths and $d_s(h) + \Dist(h, s') \le \lambda$. By construction, there is a shortest $s$--$h$ path $P$ containing only upward edges. We have
  \begin{equation*}
    \Dist(s, h) + \Dist(h, s') = d_s(h) + \Dist(h, s') \le \lambda,
  \end{equation*}
  where the equality follows from the fact that $P$ contains only upward edges and the inequality uses the definition of $h$. Then, by proposition (\ref{enum:prop-2}), a reverse BCH search from $h$ finds a shortest $s$--$h$ path, i.e., there is a shortest $s$--$h$ path~$P'$ that is an up-down path with highest-ranked vertex~$h'$ and $(s, d_s(h')) \in \SourceBucket(h')$. We have
  \begin{equation*}
    \Rank(h) \le \Rank(h') \le \Rank(h),
  \end{equation*}
  where the first inequality uses the fact that $h'$ is the highest-ranked vertex on $P'$ and the second inequality follows from $h$ being the highest-ranked vertex on all shortest $s$--$h$ paths. Thus $h' = h$ and $(s, d_s(h)) \in \SourceBucket(h)$, which completes the proof.
\end{proof}

\subsubsection*{Bucket Entry Generation.}

To exploit \cref{thm:elliptic-pruning} in practice, we need an algorithm that can efficiently check the conditions (a) and (b). Recall that with standard BCHs, we generate source bucket entries~$(s, d_s(h))$ by running a forward CH search from $s$ and inserting, for each vertex~$h$ settled, an entry~$(s, d_s(h))$ into $\SourceBucket(h)$ (the case of target bucket entries is symmetric). To check condition~(b), we need the distance~$\Dist(h, s')$ for each vertex~$h$ in the search space of the forward search. We propose the following approach.

We run a \emph{topological} forward CH search from $s$, i.e., we process vertices in topological order rather than in increasing order of distance. We prune the search at any vertex with a distance label greater than $\lambda(s, s')$ but do not apply stall-on-demand. The search stops when the priority queue becomes empty. Afterwards, we run a standard reverse CH search from $s'$. We apply stall-on-demand and stop the search as soon as the minimum key in its priority queue exceeds $\lambda(s, s')$. Finally, we need to propagate the distance labels of the reverse search down into the search space of the forward search.

We push each vertex settled during the forward search onto a stack. After the reverse search has terminated, we repeatedly pop a vertex~$u$ from the stack. For each upward edge~$(u, u')$ going out of $u$, we set $d_{s'}(u) = \min\{d_{s'}(u), \ell(u, u') + d_{s'}(u')\}$. We claim that when the stack becomes empty, we have $d_{s'}(h) = \Dist(h, s')$ for each vertex~$h$ in the search space of the forward search with $d_s(h) + \Dist(h, s') \le \lambda(s, s')$, and thus can check condition~(b).

\begin{lemma}
  \label{thm:generation-correctness}
  When the algorithm terminates, we have $d_{s'}(h) = \Dist(h, s')$ for each vertex~$h$ in the search space of the forward search with $d_s(h) + \Dist(h, s') \le \lambda(s, s')$.
\end{lemma}

\begin{proof}
  Consider one such $h$ in particular and let $w$ be the highest-ranked vertex on a shortest $h$--$s'$ path (see \cref{fig:bch-correctness-proof}). The reverse CH search is guaranteed to find a shortest $w$--$s'$ path and to set $d_{s'}(w)$ to its correct value (as shown by Geisberger et al.~\cite{GeisbergerSSV12}). All we need to show is that the propagation phase finds a shortest $h$--$w$ path.

  By construction, there is a shortest $h$--$w$ path~$R$ containing only upward edges. Since $h$ is by definition in the search space of the forward search, $R$ contains only upward edges, and the distance label of each vertex on $R$ is by definition at most $\lambda(s, s')$, all vertices on $R$ are pushed onto the stack. Since the forward search settles vertices in topological order, the stack contains the vertices in the order in which they appear on $R$. Hence, the propagation phase relaxes the edges on $R$ in reverse order and thus finds the $h$--$w$ path~$R$.
\end{proof}

It remains to check condition~(a). Consider a vertex~$h$ in the search space of the forward search and let $P$ be a shortest of the $s$--$h$ paths that contain only upward edges. Condition~(a) is violated if and only if there is an up-down $s$--$h$ path~$P'$ with at least one downward edge and $\ell(P') \le \ell(P)$; see \cref{fig:bch-correctness-proof}. We try to find such \emph{witnesses} during the propagation phase.

When we pop $h$ from the stack, we additionally look at all downward edges~$(h'', h)$ coming into $h$ and compute $\mu = \min_{(h'', h)} d_s(h'') + \ell(h'', h)$. If $\mu \le d_s(h)$, we found a witness, condition~(a) is violated, and thus we do not insert an entry into $\SourceBucket(h)$. Either way, we set $d_s(h) = \min\{d_s(h), \mu\}$. Note that we find a witness if and only if all vertices on it are contained in the search space of the forward search. Therefore, we do not necessarily discover all violations of condition~(a). However, we observed that in practice undiscovered violations are quite rare. More importantly, undiscovered violations may yield superfluous bucket entries but do not affect the correctness of the BCH searches.

\subsubsection*{Bucket Entry Removal.}

Whenever a vehicle completes a stop, we have to remove the bucket entries referring to this stop. In the following, we show how to efficiently remove the source bucket entries that refer to a stop~$s$. The case of target bucket entries is symmetric.

We initialize both a set~$R$ of reached vertices and a queue~$Q$ with the location~$l(s)$ of $s$. While $Q$ is not empty, we extract a vertex~$v$ from the queue and scan its source bucket~$\SourceBucket(v)$. When we find an entry~$(s, d_s(v))$ referring to $s$, we remove $(s, d_s(v))$ from $\SourceBucket(v)$, stop the scan, look at each upward edge~$(v, w)$ out of $v$, and insert $w$ into both $R$ and $Q$ if $w \notin R$.

The algorithm finds an entry~$(s, d_s(w)) \in \SourceBucket(w)$ if and only if there is an $s$--$w$ path~$P$ such that $P$ contains only upward edges and $(s, d_s(v)) \in \SourceBucket(v)$ for each vertex~$v$ on $P$. There would always be such a path~$P$ if we were able to guarantee to discover all violations of condition~(a). Since we cannot, we explicitly ensure that there is always such a path~$P$. Whenever we insert an entry into a source bucket~$\SourceBucket(w)$, we also insert a corresponding entry into $\SourceBucket(\Parent(w))$, where $\Parent(w)$ is the parent pointer of $w$ computed by the forward search. Our experiments will show that this almost never inserts additional bucket entries.

\subsubsection*{Implementation Details.}

Bucket entries must identify the stop they refer to. Therefore, we maintain an initially empty list of free stop IDs. Whenever we insert a stop into a vehicle's route, we take an ID from the list and assign it to the new stop. If the list is empty, we set the ID of the new stop to the maximum stop ID assigned so far plus one. Whenever we remove a stop from a route, we insert its ID into the list of free stop IDs. Bucket entries are stored and maintained in a way similar to how we handle stop attribute values.

\subsection{Shortest-Path Searches for Special Cases.}
\label{sec:special-cases}

We use BCHs to obtain most of the shortest-path distances needed to compute insertion costs. However, three special cases have to be treated separately. We discuss each of them in this section.

\subsubsection*{From Vehicles to Pickup.}

First, consider an insertion~$(\nu, r, i, j)$ with $R(\nu) = \langle s_0, \dots, s_k \rangle$ and $0 = i < k$. Here, the new pickup is inserted before the next scheduled stop on a vehicle's route. In this case, the vehicle is immediately diverted to the new pickup. To compute the cost of the insertion, we need the shortest-path distance~$\Dist(\CurrentLocation(\nu), p(r))$ from the current location~$\CurrentLocation(\nu)$ of the vehicle~$\nu$ to the pickup spot~$p(r)$. Note that our BCH searches do not find shortest paths from the vehicle's current location. Since the current location changes continuously, we cannot precompute bucket entries for it. However, the BCH searches provide us with a lower bound on the actual pickup detour.

The travel time from $s_0$ to $s_1$ via pickup spot~$p(r)$ is $\Dist(s_0, \CurrentLocation(\nu)) + \Dist(\CurrentLocation(\nu), p(r)) + \Dist(p(r), s_1)$. The inequality~$\Dist(s_0, p(r)) \le \Dist(s_0, \CurrentLocation(\nu)) + \Dist(\CurrentLocation(\nu), p(r))$ then yields a lower bound of $\Dist(s_0, p(r)) + \Dist(p(r), s_1)$ on the travel time from $s_0$ to $s_1$ via $p(r)$. Since we have source bucket entries for $s_0$ and target bucket entries for $s_1$, this lower bound can be obtained from the BCH searches. We can then compute lower bounds on the pickup detour and finally on the cost of the insertion. Only in the rare case that the latter lower bound is better than the best insertion seen so far, we have to compute the exact shortest-path distance~$\Dist(\CurrentLocation(\nu), p(r))$ by running a standard CH query.

\subsubsection*{From Last Stops to Pickup.}

Next, consider an insertion~$(\nu, r, i, j)$ with $R(\nu) = \langle s_0, \dots, s_k \rangle$ and $i = k$. Here, the new pickup is inserted after the last stop on a vehicle's route. Observe that this case also covers currently idle vehicles. To compute the cost of such insertions, we need the shortest-path distance~$\Dist(s_k, p(r))$ from the last stop~$s_k$ to the pickup spot~$p(r)$. However, our BCH searches do not find shortest paths from the last stop. The reason is that we do not generate source bucket entries for the last stop, since we cannot apply elliptic pruning in this case (the leeway is unbounded).

Instead, we defer all possible insertions~$(\nu, r, i, j)$ with $R(\nu) = \langle s_0, \dots, s_k \rangle$ and $i = k$. After having tried all possible candidate insertions~$(\nu', r, i', j')$ with $R(\nu') = \langle s'_0, \dots, s'_{k'} \rangle$ and $j' \ne k'$, we perform a reverse Dijkstra search from $p(r)$. Whenever we settle the last stop of a vehicle~$\nu$ with $R(\nu) = \langle s_0, \dots, s_k \rangle$, we check whether the insertion~$(\nu, r, k, k)$ improves the currently best insertion. Note that the detour (i.e., the increase in operation time) for each such insertion is $\Detour = \Dist(s_k, p(r)) + \StopTime + \Dist(p(r), d(r)) + \StopTime$, and thus its cost is at least $\Detour$. Therefore, we can stop the search when the sum of the minimum key~$\kappa$ in its priority queue and $\StopTime + \Dist(p(r), d(r)) + \StopTime$ is at least as large as the cost of the best insertion found so far. We can do even better by taking into account lower bounds on the violations of the wait and trip time constraint. More precisely, we can stop the search as soon as the sum
\begin{equation*}
  \begin{split}
    \kappa &+ \StopTime + \Dist(p(r), d(r)) + \StopTime\\
           &+ \WaitViolationWeight
               \cdot \max\{\kappa + \StopTime - \MaxWaitTime, 0\}\\
           &+ \TripViolationWeight
               \cdot \max\{\kappa + \StopTime + \Dist(p(r), d(r)) - \MaxTripTime(r), 0\}
  \end{split}
\end{equation*}
is at least as large as the cost of the currently best insertion. Stopping the Dijkstra search early makes it practical and fast enough for real-time applications.

\subsubsection*{From Last Stops to Dropoff.}

Lastly, consider a candidate insertion~$(\nu, r, i, j)$ with $R(\nu) = \langle s_0, \dots, s_k \rangle$ and $i < j = k$. Here, the new pickup is inserted before and the new dropoff is inserted after the last stop on a vehicle's route. To compute the cost of that insertion, we need the shortest-path distance~$\Dist(s_k, d(r))$ from the last stop~$s_k$ to the dropoff spot~$d(r)$. As discussed before, our BCH searches do not find shortest paths from the last stop. Instead, we treat this special case similarly to the previous one.

After running a reverse Dijkstra search from $p(r)$, we also run one from $d(r)$. Whenever we settle the last stop of a vehicle~$\nu$ with $R(\nu) = \langle s_0, \dots, s_k \rangle$, we check whether any insertion~$(\nu, r, i, k)$ with $i < k$ improves the best insertion seen so far. Since the cost of each such insertion is at least $\Dist(s_k, d(r)) + \StopTime$, we can stop the search when the sum of the minimum key~$\kappa$ in its priority queue and $\StopTime$ is at least as large as the cost of the currently best insertion. Again, we can do better by taking into account a lower bound on the violation of the request's trip time constraint. Then, we can stop the search as soon as the sum
\begin{equation*}
  \kappa + \StopTime + \TripViolationWeight \cdot \max\{\StopTime + \kappa - \MaxTripTime(r), 0\}
\end{equation*}
is as large as the cost of the best insertion found so far.

\subsection{Putting Everything Together.}
\label{sec:putting-everything-together}

In this section we assemble the basic LOUD algorithm from the building blocks introduced in the preceding sections. Given a ride request~$r = (p, d, \MinDepTime)$, the algorithm inserts it into any vehicle's route such that the vehicle's detour plus the violations of the soft constraints (if any) is minimized. A request is resolved in four phases, and we explain each in turn. In addition, \cref{algo:resolve-request} gives high-level pseudocode for each phase.

\begin{algorithm2e*}[tbp]
  \caption{Routine for resolving a received ride request~$r = (p, d, \MinDepTime)$.}
  \label{algo:resolve-request}

  \BeginBox
  run a CH query from pickup $p$ to dropoff $d$
      \BoxHeading{r}{KITpalegreen}{Computing Shortest-Path Distances}
  $\MaxDepTime(r) \gets \MinDepTime(r) + \MaxWaitTime$\;
  $\MaxArrTime(r) \gets \MinDepTime(r) + \alpha \cdot \Dist(p, d) + \beta$\;
  run forward and reverse BCH searches from pickup spot $p$ and dropoff spot $d$\;

  \EndBoxBeginBox{KITpalegreen}
  let $\hat{\iota} = \smash{(\hat{\nu}, r, \hat{i}, \hat{j})} \gets \bot$
      be the best insertion found so far\BoxHeading{r}{KITcyanblue}{Trying Ordinary Insertions}
  \ForEach{vehicle $\nu \in C$}{
    let $\langle s_0, \dots, s_k \rangle$ be the route of vehicle $\nu$\;
    \For{$i \gets 1$ \KwTo $k - 1$}{
      \lIf{$o(s_i) = c(\nu)$}{continue}
      try to improve $\hat{\iota}$ with insertion $(\nu, r, i, i)$\;
      \For{$j \gets i + 1$ \KwTo $k - 1$}{
        \If{$o(s_j) = c(\nu)$}{
          \If{$l(s_j) = d$}{
            try to improve $\hat{\iota}$ with insertion $(\nu, r, i, j)$\;
          }
          break\;
        }
        try to improve $\hat{\iota}$ with insertion $(\nu, r, i, j)$\;
      }
    }
  }

  \EndBoxBeginBox{KITcyanblue}
  \ForEach(\BoxHeading{f}{KITblue}{Trying Special-Case Insertions}){vehicle $\nu \in C$}{
    try to improve $\hat{\iota}$ with any insertion $(\nu, r, 0, j)$ with $0 \le j < |R(\nu)| - 1$\;
  }
  search for insertions better than $\hat{\iota}$ that insert the pickup at the end of a route\;
  search for insertions better than $\hat{\iota}$ that insert the dropoff at the end of a route\;

  \EndBox{KITblue}
  \lIf{no feasible insertion has been found}{\KwRet{$\bot$}}

  \BeginBox
  let $\langle s_0, \dots, s_k \rangle$ be the route of vehicle $\hat{\nu}$
      \BoxHeading{r}{KITlilac}{Updating Preprocessed Data}
  $\langle s'_0, \dots, s'_{i'} = p, \dots, s'_{j'} = d, \dots, s'_{k'} \rangle \gets$
      perform insertion $\hat{\iota}$\;

  \If{vehicle $\hat{\nu}$ is diverted while driving from $s_0$ to $s_1$}{
    remove source bucket entries for stop $s'_0$\;
    $l(s'_0) \gets \CurrentLocation(\hat{\nu})$\;
    $\MinDepTime(s'_0) \gets$ current point in time\;
    generate source bucket entries for stop $s'_0$\;
  }

  \If{the pickup is not inserted at an existing stop}{
    generate source and target bucket entries for stop $s'_{i'}$\;
  }
  \If{the dropoff is not inserted at an existing stop}{
    generate target bucket entries for stop $s'_{j'}$\;
    \eIf{the dropoff is inserted before the last stop}{
      generate source bucket entries for stop $s'_{j'}$\;
    }{
      generate source bucket entries for stop $s_k$\;
    }
  }

  \If{$l(s_k) \ne l(s'_{k'})$}{
    remove $\hat{\nu}$ from the list of vehicles that terminate at $l(s_k)$\;
    insert $\hat{\nu}$ into the list of vehicles that terminate at $l(s'_{k'})$\;
  }

  \EndBox{KITlilac}
  \KwRet{$\hat{\iota}$}\;
\end{algorithm2e*}

\subsubsection*{Computing Shortest-Path Distances.}

We start by computing the shortest-path distance from the pickup~$p$ to the dropoff~$d$ with a standard CH query. From this distance, we compute the latest time~$\MaxDepTime(r)$ when $r$ should be picked up as well as the latest time~$\MaxArrTime(r)$ when $r$ should be dropped off. Next, we compute all shortest-path distances that we need to calculate the costs of all \emph{ordinary} insertions, i.e., insertions~$(\nu, r, i, j)$ with $0 < i \le j < |R(\nu)| - 1$. We do this by running two forward BCH searches (from $p$ and $d$) that scan the target buckets, and two reverse BCH searches (from $p$ and $d$) that scan the source buckets.

\subsubsection*{Trying Ordinary Insertions.}

Next, we try all possible ordinary insertions. To do so, we look at the set~$C$ of vehicles that have been seen while scanning the buckets (recall that we store in each bucket entry the identifier of the vehicle to which the entry belongs). Note that vehicles that are not contained in $C$ allow no feasible ordinary insertions, and thus we do not have to consider them during this phase of the algorithm.

For each vehicle~$\nu \in C$, we enumerate all ordinary insertions that satisfy the capacity constraints, using the occupancy values~$o(\cdot)$ that we computed in \cref{sec:maintaining-feasibility}. Let $\langle s_0, \dots, s_k \rangle$ be the route of $\nu$. We loop over all pickup insertion points~$i$, $0 < i < k$, in increasing order. If the number~$o(s_i)$ of occupied seats when $\nu$ departs from $s_i$ is equal to the capacity~$c(\nu)$ of $\nu$, then all insertions~$(\nu, r, i, \cdot)$ are infeasible, and we continue with the next pickup insertion point. Otherwise, we loop over all dropoff insertion points~$j$, $i \le j < k$, in increasing order. If $o(s_j) < c(\nu)$, then the insertion~$(\nu, r, i, j)$ satisfies the capacity constraints. Otherwise, all insertions~$(\nu, r, i, \ell)$ with $\ell > j$ are infeasible, and we continue with the next pickup insertion point. The insertion with $\ell = j$ satisfies the constraints only if $d$ coincides with $s_j$.

For each insertion~$\iota$ satisfying the capacity constraints, we check whether the remaining hard constraints are also satisfied and compute the insertion cost according to \cref{eq:objective-value}. This can be done in constant time using the subroutine we introduced in \cref{sec:maintaining-feasibility}. Finally, if $\iota$ improves the best insertion~$\hat{\iota}$ found so far, we update $\hat{\iota}$ accordingly.

\subsubsection*{Trying Special-Case Insertions.}

Next, we try all possible special-case insertions, i.e., insertions whose cost depends on some shortest-path distances not computed by the BCH searches. First, we try all insertions~$(\nu, r, 0, j)$ with $0 \le j < |R(\nu)| - 1$. Such insertions insert the pickup before the next scheduled stop on a vehicle's route. Since vehicles~$\nu' \notin C$ allow no feasible insertions~$(\nu', r, 0, j)$ with $0 \le j < |R(\nu')| - 1$, it suffices to look at each vehicle~$\nu \in C$. Let $\langle s_0, \dots, s_k \rangle$ be the route of $\nu$. If $o(s_0) = c(\nu)$, then $\nu$ is currently fully occupied, and thus we cannot pick up another request before the next scheduled stop. If $o(s_0) < c(\nu)$, then we loop over all dropoff insertion points~$j$, $0 \le j < k$, terminating the loop when $o(s_j) = c(\nu)$. For each $j$, we handle the insertion~$(\nu, r, 0, j)$ as described in \cref{sec:special-cases}.

Second, we search for insertions better than $\hat{\iota}$ that insert both the pickup and the dropoff after the last stop on a vehicle's route. We do this by performing a reverse Dijkstra search from $p$, as discussed in \cref{sec:special-cases}. Finally, we search for insertions better than $\hat{\iota}$ that insert only the dropoff after the last stop on a vehicle's route. To do that, we run a reverse Dijkstra search from $d$, as described also in \cref{sec:special-cases}.

\subsubsection*{Updating Preprocessed Data.}

If we have found a feasible insertion, we need to update the preprocessed data in order to be ready to receive and resolve the next ride request. We start by actually \emph{performing} the best insertion~$\hat{\iota} = (\hat{\nu}, r, \hat{i}, \hat{j})$ into the current route~$\langle s_0, \dots, s_k \rangle$ of $\hat{\nu}$. Let $\langle s'_0, \dots, s'_{i'} = p, \dots, s'_{j'} = d, \dots, s'_{k'} \rangle$ be the route of $\hat{\nu}$ after the insertion. The $\MinDepTime$, $\MaxArrTime$, and $o$~values can be updated in time linear in the length of the route.

If $\hat{\nu}$ is diverted while driving from $s_0$ to $s_1$, we update the start~$s'_0$ of its current leg and recompute the source bucket entries for $s'_0$. (Note that there are no target bucket entries for $s'_0$ because it is the first stop on the route.) First, we remove the current source bucket entries for $s'_0$. Then, we set the location of $s'_0$ to the current location of $\hat{\nu}$, and the departure time at $s'_0$ to the current point in time. Finally, we generate new source bucket entries for stop~$s'_0$.

Moreover, we generate source and target bucket entries for the stop~$s'_{i'}$ at which the pickup is made unless the pickup is inserted at an existing stop. Likewise, we generate target bucket entries for the stop~$s'_{j'}$ at which the dropoff is made unless the dropoff is inserted at an existing stop. If the dropoff is inserted before the last stop, we also generate source bucket entries for $s'_{j'}$. Otherwise, we generate source bucket entries for the stop~$s_k$ that was at the very end of the route before performing the insertion. (Note that whenever a vehicle reaches the next scheduled stop on its route, we remove the target bucket entries for this stop, and the source bucket entries for the preceding stop.)

It remains to update one more data structure. For each vertex~$v$, we maintain a list of vehicles that terminate at $v$, i.e., whose currently last stop is made at $v$. Whenever the reverse Dijkstra searches from $p$ and $d$ settle a vertex~$v$, they retrieve the last stops at $v$ with these lists. Therefore, we remove $\hat{\nu}$ from the list of vehicles terminating at $l(s_k)$, and we insert $\hat{\nu}$ into the list of vehicles terminating at $l(s'_{k'})$.

\subsection{Extensions.}
\label{sec:extensions}

This section shows how LOUD can be extended to meet additional requirements of real-world production systems. We explain each extension in turn, but they can be combined in an actual implementation. Our implementation supports all of them.

\subsubsection*{Edge-Based Stops.}

Up to now, we have assumed that stops are made at vertices (i.e., intersections). In real-world applications, however, stops are made anywhere along edges (i.e., road segments). Fortunately, LOUD can be easily extended to work with edge-based stops, following the approach proposed by Delling et al.~\cite{DellingGPW17}.

Consider a stop~$s$ along an edge~$e = (v, w)$ with a real-valued offset~$o \in [0, 1]$. To run a forward search (whether it is a Dijkstra, CH, or BCH search) from $s$, we start from the \emph{head} vertex~$w$ and initialize the distance label~$d_w(w)$ to $(1 - o) \cdot \ell(e)$ rather than zero. Likewise, to run a reverse Dijkstra, CH, or BCH search from $s$, we start from the \emph{tail} vertex~$v$ and initialize the distance label~$d_v(v)$ to $o \cdot \ell(e)$. The special case where source and target are located on the same edge is treated explicitly.

\subsubsection*{Path Retrieval.}

In real-world applications, one is often interested not only in the best insertion~$(\nu, r, i, j)$ but also in the descriptions of the paths from stop~$s_i$ to the pickup spot~$p(r)$, from $p(r)$ to stop~$s_{i + 1}$, from stop~$s_j$ to the dropoff spot~$d(r)$, and from $d(r)$ to stop~$s_{j + 1}$. By maintaining a parent pointer for each vertex, the \mbox{Dijkstra} searches can retrieve complete path descriptions, and the CH searches can retrieve descriptions potentially containing shortcuts. The latter can be unpacked into complete descriptions in time linear in the number of edges on the unpacked path~\cite{GeisbergerSSV12}.

Now, consider a path~$\langle s, \dots, h, \dots, s' \rangle$ found by a forward BCH search. The case of a reverse BCH search is symmetric. Let $h$ be the highest-ranked vertex on the path. Since the $s$--$h$ path is found by a forward CH search, its description can be retrieved as discussed above. The $h$--$s'$ path, however, is hidden behind the target bucket entry~$(s', d_{s'}(h)) \in \TargetBucket(h)$. Therefore, it remains to retrieve the path description that corresponds to a target bucket entry.

When we generate target bucket entries for $s'$, we could explicitly store the search space of $s'$ as a rooted tree~$T_{s'}$. To retrieve the description of the $h$--$s'$ path, we would traverse the path in $T_{s'}$ from $h$ to $s'$. Note, however, that to find a best insertion, we need no parent information. That is, $T_{s'}$ is only needed when we insert a new stop immediately before $s'$, which may never be the case. Since it seems wasteful to build a tree that may never be used, we instead retrieve the path description corresponding to a target bucket entry~$(s', d_{s'}(h))$ by running a reverse CH search (from $s'$ to $h$) when needed.

\subsubsection*{Handling Traffic.}

Today's ridesharing services have to be able to quickly update the routing graph whenever new traffic information is available. On large-scale road networks, however, CH preprocessing is not fast enough to incorporate a continuous stream of traffic information. Hence, we propose combining LOUD with \emph{customizable} contraction hierarchies (CCHs)~\cite{DibbeltSW16}, a CH variant that can incorporate new metrics in few seconds. As a customizable contraction hierarchy \emph{is a} contraction hierarchy, LOUD can be used as is with CCHs, without the need for further modifications.

We can do better by replacing the Dijkstra-based CH searches with elimination tree searches, a query algorithm tailored to CCHs. Elimination tree searches tend to be faster than Dijkstra-based searches for point-to-point queries, however, they have one drawback. Since they do not process vertices in increasing order of distance, it is not clear how to terminate them early. This is an issue because the Dijkstra-based CH searches during bucket entry generation have a tight stopping criterion. However, we observe that we can turn \emph{stopping} criteria for Dijkstra-based CH searches into \emph{pruning} criteria for elimination tree searches.

During bucket entry generation, the Dijkstra-based CH searches stop as soon as they settle a vertex whose distance label exceeds the leeway. We cannot \emph{stop} an elimination tree search at such a vertex~$v$. However, we can \emph{prune} the search at $v$, i.e., we do not relax edges out of $v$. As shown by Buchhold et al.~\cite{BuchholdSW19}, the edge relaxations are the time-consuming part, whereas the time spent on elimination tree traversal is negligible.

Note that elimination tree searches even simplify bucket entry generation. In \cref{sec:elliptic-pruning}, we have introduced special \emph{topological} CH searches. Since elimination tree searches process vertices in ascending rank order, and the rank order is a topological order, each standard elimination tree search is already a topological search.

There is, however, a potential pitfall associated with customization. Recall that to remove bucket entries for a stop~$s$, we essentially simulate a CH search from $s$ to find the buckets that contain entries referring to $s$. This requires that the topology of the hierarchy does not change between generation and removal of the bucket entries for $s$. Fortunately, CCHs compute a metric-independent contraction order during a preprocessing step, i.e., customization does not affect the order. Thus, when using \emph{basic} CCH customization~\cite{DibbeltSW16}, the topology does not change, and we can safely update the edge costs between bucket entry generation and removal.

For even smaller search spaces, we can apply a more sophisticated customization algorithm (\emph{perfect} customization~\cite{DibbeltSW16}). This additionally removes superfluous edges from the hierarchy. Therefore, although the contraction order remains the same, the topology of the hierarchy may change. Hence, when using perfect customization, we have to clear and rebuild the source and target buckets after each customization step.

\subsubsection*{Other Objective Functions.}

Our precise objective function is taken from the popular transport simulation MATSim~\cite{HorniNA16, BischoffMN17}, and can be parameterized as discussed in \cref{sec:problem-statement}. We stress, however, that LOUD is not restricted to this objective function but can work with other functions as well. Note that elliptic pruning (and therefore bucket entry generation) does not depend on the objective function, only on the hard constraints for requests already matched to a vehicle. Hence, it will perform similarly for \emph{any} objective function. The only ingredients that depend on the actual objective function are the stopping criteria for the reverse Dijkstra searches from the received pickup and dropoff spot, respectively.

\section{Experiments.}
\label{sec:experiments}

This section presents a thorough experimental evaluation of LOUD on the state-of-the-art Open Berlin Scenario~\cite{ZiemkeKN19} and a second, even larger benchmark instance, including a comparison to related work. We also integrate LOUD into a transport simulation software.

\subsection{Experimental Setup.}

Our source code is written in C++17 and compiled with the GNU compiler~9.3 using optimization level~3. We use 4-heaps~\cite{Johnson75} as priority queues. To ensure a correct implementation, we make extensive use of assertions (disabled during measurements). Our benchmark machine runs openSUSE Leap~15.2 (kernel~5.3.18), and has 192\,GiB of DDR4-2666 RAM and two Intel Xeon Gold~6144 CPUs, each with eight cores clocked at 3.50\,GHz and $8 \times 64$\,KiB of L1, $8 \times 1$\,MiB of L2, and 24.75\,MiB of shared L3~cache. Note that we consider only single-core implementations.

\subsubsection*{Inputs.}

Our main benchmark instances are taken from the Open Berlin Scenario~\cite{ZiemkeKN19}, a publicly available transport simulation scenario for the Berlin metropolitan area implemented in MATSim~\cite{HorniNA16}. The transport simulation MATSim works in iterations, with each iteration simulating the movement of the given population (including departure time, route, mode and destination choice) and outputting each inhabitant's 24-hour travel pattern. Over the course of iterations, the activity-travel patterns become more and more realistic.

To obtain a set of realistic ride requests, we build on the Open Berlin Scenario 5.5 with demand-responsive transport (DRT). By default, only a few trips use DRT. Therefore, we change three parameters. We halve the DRT fare per kilometer from 35 to 18~cents, halve the minimum DRT fare per trip from 2 to 1~euro, and double the daily cost per private car from 5.30 to 10.60~euros. This primarily replaces private-car trips by DRT trips.

The Open Berlin Scenario has been published in two versions. The 1\,\% scenario simulates 1\,\% of all adults living in Berlin and Brandenburg, while the 10\,\% scenario simulates 10\,\% of them. For our benchmark instance \emph{Berlin-1pct}, we take all DRT requests from the 500th~iteration of the 1\,\% scenario (500 is the number of iterations recommended for realistic travel patterns). For our instance \emph{Berlin-10pct}, we take all DRT requests from the 250th iteration of the 10\,\% scenario (since one iteration takes more than four hours, performing 500 is not feasible). Both instances take the network from the Open Berlin Scenario, which builds on OpenStreetMap.

To evaluate LOUD on even larger instances, we build two additional instances that comprise the Rhine-Ruhr area, the largest metropolitan area in Germany. The construction is guided by the Open Berlin Scenario. We start by taking the network from OpenStreetMap. Besides all roads in the city of Berlin, the Open Berlin Scenario includes all main roads in Brandenburg (the state that surrounds Berlin). For our Rhine-Ruhr instances, we therefore take all roads in the Rhine-Ruhr metropolitan area (as defined by the Landesentwicklungsplan NRW from 1995) and all main roads in the surrounding state of North Rhine-Westphalia.

As for the Open Berlin Scenario, we build a sparser instance \emph{Ruhr-1pct} and a denser instance \emph{Ruhr-10pct}. Since the population in the Rhine-Ruhr area is roughly three times larger than in the Berlin area, we scale the numbers of vehicles and requests for Berlin-1pct (Berlin-10pct) by a factor of three to obtain the numbers for Ruhr-1pct (Ruhr-10pct). The initial locations of the vehicles are uniformly distributed in the Rhine-Ruhr area (no vehicle starts in the surroundings). As the population density correlates with the density of the graph, vehicles tend to start in densely populated areas.

We choose the pickup and dropoff spot for a request as follows. First, we choose the pickup spot uniformly at random from the Rhine-Ruhr area. Next, we draw the trip duration from a geometric distribution with probability parameter~$p = 1 / (\mu + 1)$. Finally, we run Dijkstra's algorithm from the pickup spot until we settle a vertex whose distance label is greater than or equal to the trip duration drawn before \emph{and} that vertex is contained in the Rhine-Ruhr area. The expected trip duration~$\mu$ is set to the average trip duration on the corresponding Berlin instance (12~minutes on Berlin-1pct and 11~minutes on Berlin-10pct).

The earliest departure time for a request is drawn according to the distribution of the earliest departure times on the Berlin instances. More precisely, we group the departure times on the Berlin instances into five-minute bins~$b_i$. To choose the departure time for a request on the Rhine-Ruhr instances, we first draw a bin~$B$ from the discrete distribution determined by the probability function~$\Pr[B = b_i] = |b_i| / \sum_{b_j} |b_j|$, and then choose the departure time uniformly at random from the interval corresponding to $B$. Key figures of the Berlin and Rhine-Ruhr instances are shown in \cref{tab:instances}.

\begin{table}[tb]
  \caption{Key figures of our benchmark instances.}
  \label{tab:instances}
  \begin{tabular*}{\linewidth}{
    l@{\extracolsep{\fill}}S@{\extracolsep{\fill}}S@{\extracolsep{\fill}}S@{\extracolsep{\fill}}S}
  \toprule
  input & {$|V|$} & {$|E|$} & {veh} & {req} \\
  \midrule
  Berlin-1pct  &  73689 & 159039 &  1000 &  16569 \\
  Berlin-10pct &  73689 & 159039 & 10000 & 149185 \\
  \addlinespace
  Ruhr-1pct    & 394049 & 840587 &  3000 &  49707 \\
  Ruhr-10pct   & 394049 & 840587 & 30000 & 447555 \\
  \bottomrule
\end{tabular*}
\end{table}

\subsubsection*{Methodology.}

We implemented a discrete-event simulation that simulates a given set of vehicles servicing a given set of requests. The simulation maintains each vehicle's current state (out of service, idling, driving, or stopping) and an addressable priority queue of pending events. Each event happens at some scheduled point in time and may generate a new event in the future. We repeatedly extract the next event from the queue and process it. The transport simulation stops as soon as the event queue becomes empty.

For each ride request~$r$ in the input, we process a \emph{request receipt event} at $\MinDepTime(r)$. To do so, we match request~$r$ to some vehicle~$\nu$. If $\nu$ is currently idling, we set its state to driving and insert a vehicle arrival event at $\Now + \Dist(\CurrentLocation(\nu), p(r))$ into the queue, where $\Now$ is the current point in time. If vehicle~$\nu$ is currently driving and $r$ is inserted before the next scheduled stop, we update the scheduled time of $\nu$'s existing vehicle arrival event to $\Now + \Dist(\CurrentLocation(\nu), p(r))$.

For each vehicle~$\nu$ in the input, we process a \emph{vehicle startup event} at $\MinServTime(\nu)$ and a \emph{vehicle shutdown event} at $\MaxServTime(\nu)$. To process the former, we check whether there are already any requests matched to $\nu$. If so, we set $\nu$'s state to driving and insert a vehicle arrival event into the queue. Otherwise, we set the state to idling and generate no new event. To process the vehicle shutdown event, we set $\nu$'s state to out of service and notify the dispatching algorithm about the vehicle shutdown. Note that all request receipt, vehicle startup and vehicle shutdown events are known in advance and form the initial content of the event queue.

Whenever a vehicle~$\nu$ reaches a stop, we process a \emph{vehicle arrival event}. To do so, we set $\nu$'s state to stopping and add a vehicle departure event at $\Now + \StopTime$ to the queue. Moreover, we notify the dispatching algorithm about the vehicle arrival so that $\nu$'s route (and preprocessed data) can be updated. Finally, whenever a vehicle~$\nu$ is ready to depart from a stop, we process a \emph{vehicle departure event}. To do so, we check whether there are currently any ride requests matched to $\nu$. If so, we set its state to driving and insert a vehicle arrival event into the queue. Otherwise, we set the state to idling and generate no new event.

\begin{table*}[tb]
  \caption{Bucket entry generation on various benchmark instances with standard and customizable CHs. We report the total number of vertices~$v$ in the search space of a newly inserted stop~$s$ with neighboring stop~$s'$. We also report those that are the highest-ranked vertex on all shortest paths between $s$ and $v$ (i.e., satisfy condition~(a)), those that lie inside the shortest-path ellipse around $s$ and $s'$ (i.e., satisfy condition~(b)), and those that satisfy both conditions. Moreover, we report the number of bucket entries inserted, the running time for the search from the new stop, the search from its neighbor, the propagation of distance labels, and the total running time.}
  \label{tab:bucket-entry-generation}
  \begin{tabular*}{\linewidth}{
    l@{\extracolsep{\fill}}l@{\extracolsep{\fill}}r@{\extracolsep{\fill}}r@{\extracolsep{\fill}}
    r@{\extracolsep{\fill}}r@{\extracolsep{\fill}}r@{\extracolsep{\fill}}r@{\extracolsep{\fill}}
    r@{\extracolsep{\fill}}r@{\extracolsep{\fill}}r}
  \toprule
  & & \multicolumn{4}{c}{\makebox[0pt]{\#\,vertices in search space}} & \# & \multicolumn{4}{c}{\makebox[0pt]{running time [\si{\micro\second}]}} \\
  \cmidrule{3-6}\cmidrule{8-11}
  input & CH & total & \hspace{-\tabcolsep}highest & \hspace{-\tabcolsep}in ellipse & both & entries & stop & neigh & prop & total \\
  \midrule
  Berlin & std  & 210.37 &  54.54 & 16.90 &  9.87 &  9.87 & 4.33 & 3.61 & 2.24 & 10.17 \\
  1pct   & cust & 186.63 & 136.63 & 15.50 & 12.49 & 12.50 & 2.66 & 2.85 & 2.21 &  7.72 \\
  \addlinespace
  Berlin & std  & 210.64 &  54.66 & 14.05 &  8.72 &  8.73 & 4.03 & 3.35 & 1.99 &  9.37 \\
  10pct  & cust & 186.76 & 136.35 & 13.20 & 10.84 & 10.84 & 2.49 & 2.66 & 1.95 &  7.10 \\
  \addlinespace
  Ruhr   & std  & 241.09 &  54.38 & 15.67 &  8.70 &  8.71 & 3.73 & 3.35 & 2.32 &  9.40 \\
  1pct   & cust & 228.91 & 165.20 & 13.67 & 11.42 & 11.42 & 3.00 & 3.22 & 3.05 &  9.26 \\
  \addlinespace
  Ruhr   & std  & 241.58 &  54.46 & 14.25 &  8.43 &  8.43 & 3.46 & 3.11 & 2.09 &  8.66 \\
  10pct  & cust & 228.91 & 165.26 & 12.88 & 10.90 & 10.90 & 2.77 & 2.95 & 2.53 &  8.25 \\
  \bottomrule
\end{tabular*}
\end{table*}

\begin{table*}[tb]
  \caption{Time (in microseconds) for BCH searches and bucket entry removal on various benchmark instances with standard and customizable CHs. We also report the number of vertices and bucket entries visited during a BCH search and while removing bucket entries referring to a completed stop.}
  \label{tab:bch-searches}
  \begin{tabular*}{\linewidth}{
    l@{\extracolsep{\fill}}l@{\extracolsep{\fill}}S@{\extracolsep{\fill}}S@{\extracolsep{\fill}}
    S@{\extracolsep{\fill}}S@{\extracolsep{\fill}}S@{\extracolsep{\fill}}S}
  \toprule
  & & \multicolumn{3}{c}{\makebox[0pt]{BCH searches}} & \multicolumn{3}{c}{\makebox[0pt]{bucket entry removal}} \\
  \cmidrule{3-5}\cmidrule{6-8}
  input & CH & {\#\,vertices} & {\#\,entries} & {time} & {\#\,vertices} & {\#\,entries} & {time} \\
  \midrule
  Berlin & std  &  62.87 &   564.16 &  14.94 & 25.72 &  149.54 & 1.21 \\
  1pct   & cust & 186.65 &  1331.91 &  16.24 & 46.16 &  293.23 & 1.68 \\
  \addlinespace
  Berlin & std  &  62.94 &  3994.05 &  35.55 & 23.57 &  905.88 & 1.73 \\
  10pct  & cust & 186.66 &  9149.83 &  52.88 & 42.22 & 1764.32 & 2.61 \\
  \addlinespace
  Ruhr   & std  &  65.80 &  1191.24 &  19.84 & 19.62 &  140.86 & 1.85 \\
  1pct   & cust & 229.15 &  4031.66 &  33.09 & 41.06 &  446.38 & 2.93 \\
  \addlinespace
  Ruhr   & std  &  65.85 &  7953.42 &  65.80 & 18.61 &  820.10 & 2.67 \\
  10pct  & cust & 229.18 & 25475.05 & 133.56 & 38.38 & 2540.69 & 4.90 \\
  \bottomrule
\end{tabular*}
\end{table*}

\subsubsection*{Parameters.}

We take the default model parameters from MATSim. The stop time~$\StopTime$ is set to \SI{1}{\minute}, the maximum wait time~$\MaxWaitTime$ to \SI{5}{\minute}, the maximum trip time model parameters~$\alpha$ and $\beta$ to 1.7 and $\SI{2}{\minute}$, the wait time violation weight~$\WaitViolationWeight$ to 1, and finally the trip time violation weight~$\TripViolationWeight$ to 10.

CH preprocessing is taken from the open-source library RoutingKit\footnote{\url{https://github.com/RoutingKit/RoutingKit}}. We use the partitioning algorithm Inertial Flow~\cite{SchildS15} to compute a CCH order, with the balance parameter~$b$ set to 0.3. CH preprocessing and CCH order computation take less than one second each on the Berlin network. On the Rhine-Ruhr network, the former takes 4~seconds and the latter takes 6~seconds. For smaller search spaces, we apply the more sophisticated perfect CCH customization algorithm~\cite{DibbeltSW16}.

\begin{table*}[tb]
  \caption{Performance of resolving ride requests on various benchmark instances with standard and customizable CHs. We report the time to compute the shortest direct path from the pickup to the dropoff spot, the time for the BCH searches, the time to try all ordinary candidate insertions, the time to treat the special cases (pickup before the next stop, pickup after the last stop, and dropoff after the last stop), the time to update the preprocessed data (including bucket entry generation), and the total running time. All running times are given in microseconds. In addition, we report the size of the superset~$C$ of promising candidate vehicles.}
  \label{tab:request-resolution-loud}
  \begin{tabular*}{\linewidth}{
    l@{\extracolsep{\fill}}l@{\extracolsep{\fill}}S@{\extracolsep{\fill}}S@{\extracolsep{\fill}}
    S@{\extracolsep{\fill}}S@{\extracolsep{\fill}}S@{\extracolsep{\fill}}S@{\extracolsep{\fill}}
    S@{\extracolsep{\fill}}S@{\extracolsep{\fill}}S}
  \toprule
  & & & & \multicolumn{2}{c}{\makebox[0pt]{ordinary}} & \multicolumn{3}{c}{\makebox[0pt]{special insertions}} & & \\
  \cmidrule{7-9}
  & & & & \multicolumn{2}{c}{\makebox[0pt]{insertions}} & {pickup} & \hspace{-\tabcolsep}pickup & {dropoff} & & \\
  \cmidrule{5-6}
  input & CH & {direct} & {BCH} & {$|C|$} & {time} & {at beg} & \hspace{-\tabcolsep}at end & {at end} & {upd} & {total} \\
  \midrule
  Berlin & std  & 11.03 &  60.76 &  48 &  1.70 &  9.77 &  9.62 &  555.71 & 45.02 &  693.60 \\
  1pct   & cust &  8.35 &  65.99 &  48 &  1.71 &  8.84 &  9.67 &  562.30 & 35.12 &  691.99 \\
  \addlinespace
  Berlin & std  & 10.73 & 143.29 & 277 & 20.89 & 21.86 &  5.37 &  379.81 & 42.07 &  624.02 \\
  10pct  & cust &  8.02 & 213.88 & 280 & 20.79 & 20.77 &  5.20 &  369.84 & 33.46 &  671.96 \\
  \addlinespace
  Ruhr   & std  & 10.98 &  80.56 & 118 &  4.93 & 25.15 & 34.29 & 3308.34 & 42.67 & 3506.91 \\
  1pct   & cust &  9.97 & 134.04 & 117 &  5.25 & 28.37 & 34.92 & 3376.65 & 41.86 & 3631.07 \\
  \addlinespace
  Ruhr   & std  & 10.18 & 264.78 & 666 & 50.71 & 70.84 & 12.75 & 1977.47 & 40.77 & 2427.51 \\
  10pct  & cust &  8.96 & 536.08 & 661 & 51.40 & 74.00 & 12.44 & 2019.25 & 39.88 & 2742.00 \\
  \bottomrule
\end{tabular*}
\end{table*}

\subsection{Elliptic Pruning.}

We start by evaluating the effectiveness and efficiency of elliptic pruning. \Cref{tab:bucket-entry-generation} shows the reduction in search-space size achieved by conditions~(a) and (b) from \cref{thm:elliptic-pruning}. The average unpruned CH search space contains roughly 210~vertices on the Berlin instances and 240~vertices on the Ruhr instances. Only 25\,\% of them satisfy condition~(a), and even less than 10\,\% satisfy condition~(b). When combined, they decrease the average search-space size (and thus the number of bucket entries) by a factor of more than 20. With CCHs, condition~(a) prunes significantly less vertices. However, as condition~(b) still prunes more than 90\,\% of the vertices, the number of bucket entries is about the same as with standard CHs. The time to generate (source or target) bucket entries for a new stop is divided roughly equally between the search from the new stop, the search from its neighbor, and the propagation of the distance labels of the latter search into the search space of the former search.

\Cref{tab:bch-searches} shows the performance of BCH searches and bucket entry removal. Due to elliptic pruning, BCH searches scan relatively few bucket entries, and are therefore very fast. On Berlin-1pct, a BCH search takes merely 15~microseconds. On Berlin-10pct, where we have 10~times more vehicles and 9~times more ride requests, the running time doubles with standard CHs, and triples with CCHs. On Ruhr-10pct, our largest and densest instance, a BCH search takes 65~microseconds with standard CHs and twice as much with CCHs. Since we need four BCH searches per ride request, this makes at most half a millisecond, fast enough for interactive applications. Taking merely a few microseconds, the time spent on bucket entry removal is negligible.

\subsection{Resolving Ride Requests.}

We next evaluate the performance of the matching algorithm. \Cref{tab:request-resolution-loud} reports the time for each of its phases. Recall that LOUD tries only ordinary insertions into vehicles that have been seen during the BCH searches. We observe that this (exact) filter works very well, with less than 5\,\% of the vehicles passing through in all cases. Consequently, it takes only a few microseconds to try all ordinary insertions. Note that the search for special-case insertions that insert the pickup before and the dropoff after the last stop on a vehicle's route takes up the largest fraction of the total time (60\,\% on Berlin-10pct, 80\,\% on the sparser Berlin-1pct, and even almost 95\,\% on Ruhr-1pct). Interestingly, the total time on Berlin is always between 600 and 700~microseconds, although it is divided differently between the phases depending on the sparsity of the vehicles and ride requests. On the Rhine-Ruhr benchmark instances, we observe that the running times are around 3~milliseconds.

\Cref{tab:special-case-treatments} reports detailed statistics about the special-case treatments. Recall that LOUD discards as many insertions before the next scheduled stop as possible using cheap lower bounds on the pickup detour, in order to avoid costly extra CH queries. We observe that these lower bounds work very well. On average, we only need a single extra CH query per ride request on Berlin, and roughly two extra queries on the Rhine-Ruhr instances.

\begin{table*}[tb]
  \caption{Detailed statistics about the special-case treatments on various benchmark instances with standard and customizable CHs. For each special-case treatment, we report the number of insertions tried and the running time (in microseconds). For handling pickups before the next stop, we additionally report the number of CH queries needed per ride request. For handling pickups and dropoffs after the last stop, we additionally report the number of last stops visited during the reverse Dijkstra searches from the pickup and dropoff spot, respectively.}
  \label{tab:special-case-treatments}
  \begin{tabular*}{\linewidth}{
    l@{\extracolsep{\fill}}l@{\extracolsep{\fill}}S@{\extracolsep{\fill}}S@{\extracolsep{\fill}}
    S@{\extracolsep{\fill}}S@{\extracolsep{\fill}}S@{\extracolsep{\fill}}S@{\extracolsep{\fill}}
    S@{\extracolsep{\fill}}S@{\extracolsep{\fill}}S}
  \toprule
  & & \multicolumn{3}{c}{\makebox[0pt]{pickup at beginning}} & \multicolumn{3}{c}{\makebox[0pt]{pickup at end}} & \multicolumn{3}{c}{\makebox[0pt]{dropoff at end}} \\
  \cmidrule{3-5}\cmidrule{6-8}\cmidrule{9-11}
  input & CH & {inserts} & {\hspace{-\tabcolsep}queries} & {time} & {stops} & {\hspace{-\tabcolsep}inserts} & {time} & {stops} & {inserts} & {time} \\
  \midrule
  Berlin & std  &   69.68 & 0.80 &  9.77 & 1.54 & 1.54 &  9.62 &  120.90 &  18.07 &  555.71 \\
  1pct   & cust &   70.38 & 0.79 &  8.84 & 1.54 & 1.54 &  9.67 &  120.90 &  17.86 &  562.30 \\
  \addlinespace
  Berlin & std  &  582.86 & 0.80 & 21.86 & 3.88 & 3.88 &  5.37 &  731.01 & 100.86 &  379.81 \\
  10pct  & cust &  585.56 & 0.80 & 20.77 & 3.88 & 3.88 &  5.20 &  731.01 &  99.38 &  369.84 \\
  \addlinespace
  Ruhr   & std  &  184.45 & 1.76 & 25.15 & 1.74 & 1.74 & 34.29 &  302.15 &  26.27 & 3308.34 \\
  1pct   & cust &  183.48 & 1.76 & 28.37 & 1.74 & 1.74 & 34.92 &  302.15 &  25.86 & 3376.65 \\
  \addlinespace
  Ruhr   & std  & 1329.64 & 2.57 & 70.84 & 2.42 & 2.42 & 12.75 & 1795.76 & 111.29 & 1977.47 \\
  10pct  & cust & 1317.19 & 2.59 & 74.00 & 2.42 & 2.42 & 12.44 & 1795.76 & 109.18 & 2019.25 \\
  \bottomrule
\end{tabular*}
\end{table*}

\subsection{Comparison to Related Work.}

\begin{table*}[!tb]
  \caption{Performance of resolving ride requests on various benchmark instances with the heuristic MATSim algorithm and its exact variant. We report the time for the filtering phase, the search to the pickup, the search from the pickup, the search to the dropoff, the search from the dropoff, the evaluation phase, and the total running time. All running times are given in milliseconds. Moreover, we report the number of insertions tried during the filtering phase, as well as the number of filtered insertions.}
  \label{tab:request-resolution-matsim}
  \begin{tabular*}{\linewidth}{
    l@{\extracolsep{\fill}}l@{\extracolsep{\fill}}S@{\extracolsep{\fill}}S@{\extracolsep{\fill}}
    S@{\extracolsep{\fill}}S@{\extracolsep{\fill}}S@{\extracolsep{\fill}}S@{\extracolsep{\fill}}
    S@{\extracolsep{\fill}}S@{\extracolsep{\fill}}S}
  \toprule
  & & \multicolumn{3}{c}{\makebox[0pt]{geometric filtering}} & \multicolumn{4}{c}{\makebox[0pt]{Dijkstra searches}} & \multicolumn{1}{c}{eval} & \\
  \cmidrule{3-5}\cmidrule{6-9}\cmidrule{10-10}
  input & var & {tried} & {filtered} & {time} & {to $p$} & {\hspace{-\tabcolsep}from $p$} & {to $d$} & {\hspace{-\tabcolsep}from $d$} & {time} & {total} \\
  \midrule
  Berlin & heu &  1811 &   101 &  0.26 &  3.53 &  3.45 &  3.57 &  2.94 & 0.01 &  13.75 \\
  1pct   & ex  &  1811 &  1354 &  0.30 &  5.03 &  4.69 &  4.61 &  4.59 & 0.05 &  19.29 \\
  \addlinespace
  Berlin & heu & 18006 &   386 &  2.26 &  4.00 &  4.07 &  4.10 &  3.73 & 0.03 &  18.18 \\
  10pct  & ex  & 18009 & 12708 &  3.24 &  5.08 &  4.84 &  4.68 &  4.80 & 0.43 &  23.07 \\
  \addlinespace
  Ruhr   & heu &  5706 &    53 &  1.23 & 14.41 & 17.83 & 19.52 &  9.17 & 0.01 &  62.16 \\
  1pct   & ex  &  5859 &  3366 &  1.45 & 50.68 & 50.12 & 49.54 & 49.33 & 0.38 & 201.49 \\
  \addlinespace
  Ruhr   & heu & 50620 &   126 & 10.63 & 20.36 & 27.06 & 31.76 & 18.28 & 0.04 & 108.14 \\
  10pct  & ex  & 52474 & 32374 & 14.53 & 50.46 & 52.21 & 50.51 & 51.66 & 3.69 & 223.05 \\
  \bottomrule
\end{tabular*}
\end{table*}

Comparing running times is often difficult, due to different machines, benchmark instances, and programming skills. In addition, objectives and constraints in dynamic ridesharing come in a wide variety. For a fair comparison, we carefully reimplemented one competitor and run it on the same machine and instances. We choose the dispatching algorithm in MATSim for various reasons.

First, MATSim uses exactly the same problem formulation. Second, since MATSim is actually used in industry and academia, the comparison of LOUD to MATSim is of particular practical relevance. Third, since the code of MATSim is publicly available, there are no unclear implementation details (which is not the case for the other competitors). Fourth, the running times reported by the algorithms mentioned in \cref{sec:introduction} are roughly similar. On a benchmark instance comparable to Berlin-10pct, the algorithm by Huang et al.~\cite{HuangBJW14} takes between 10 and 100~milliseconds to process a ride request. For their simulated-annealing algorithm, Jung et al.~\cite{JungJP16} report running times of 174--257~milliseconds per request (on a much smaller instance). Unfortunately, \mbox{T-Share}~\cite{MaZW13} does not report any absolute running times. Our MATSim reimplementation takes 14 and 18~milliseconds per request on Berlin-1pct and Berlin-10pct, respectively; see \cref{tab:request-resolution-matsim} for further details. Note that this is 15~times faster than the official MATSim implementation, which is written in Java.

\Cref{tab:comparison} compares LOUD to the dispatching algorithm in MATSim. Besides a reimplementation of the original heuristic algorithm (MATSim-h), we also consider an exact variant (MATSim-e). Recall that the filtering phase tries all possible insertions into each vehicle's route, where all needed detours are estimated using geometric distances. More precisely, the travel time between any two vertices is given by $(\DistScalingParam \cdot \mu) / (\SpdScalingParam \cdot \VehicleSpd)$, where $\mu$ is the straight-line distance, $\VehicleSpd$ is the estimated vehicle speed, and $\DistScalingParam$ and $\SpdScalingParam$ are parameters. MATSim-h (in accordance with the official code) sets the parameters~$(\VehicleSpd, \DistScalingParam, \SpdScalingParam)$ to $(\SI{30}{\kilo\metre\per\hour}, 1.3, 1.5)$. MATSim-e sets $\VehicleSpd$ to the maximum travel speed that occurs in the network, and both $\DistScalingParam$ and $\SpdScalingParam$ to 1.

\begin{table*}[tb]
  \caption{Comparison of LOUD to the heuristic MATSim algorithm (and its exact variant). We report the average running time per request and statistics about the solution quality. For requests, we report the average and 95th percentile of the wait times, and the average ride and trip time. For vehicles, we report the average time spent driving empty, spent driving occupied, spent picking up or dropping off riders, and the average operation time.}
  \label{tab:comparison}
  \begin{tabular*}{\linewidth}{
    l@{\extracolsep{\fill}}l@{\extracolsep{\fill}}r@{\extracolsep{\fill}}r@{\extracolsep{\fill}}
    r@{\extracolsep{\fill}}r@{\extracolsep{\fill}}r@{\extracolsep{\fill}}r@{\extracolsep{\fill}}
    r@{\extracolsep{\fill}}r@{\extracolsep{\fill}}r}
  \toprule
  & & & \multicolumn{4}{c}{\makebox[0pt]{request statistics [m:s]}} & \multicolumn{4}{c}{\makebox[0pt]{vehicle statistics [h:m]}} \\
  \cmidrule{4-7}\cmidrule{8-11}
  & & time & \multicolumn{2}{c}{wait} & ride & trip & empty & occ & stop & op \\
  \cmidrule{4-5}
  instance & algorithm & [ms] & avg & \hspace{-\tabcolsep}95\,\%ile & & & & & & \\
  \midrule
  Berlin & MATSim-h &  13.76 & 4:11 &  8:21 & 14:11 & 18:22 & 0:35 & 3:19 & 0:33 & 4:27 \\
  1pct   & MATSim-e &  19.29 & 4:12 &  8:20 & 14:11 & 18:23 & 0:36 & 3:19 & 0:33 & 4:28 \\
         & LOUD-CH  &   0.70 & 4:12 &  8:20 & 14:11 & 18:23 & 0:36 & 3:19 & 0:33 & 4:28 \\
         & LOUD-CCH &   0.70 & 4:12 &  8:20 & 14:11 & 18:23 & 0:36 & 3:19 & 0:33 & 4:28 \\
  \addlinespace
  Berlin & MATSim-h &  18.18 & 3:45 &  8:21 & 14:52 & 18:36 & 0:14 & 2:31 & 0:29 & 3:14 \\
  10pct  & MATSim-e &  23.07 & 3:47 &  8:13 & 14:51 & 18:38 & 0:13 & 2:31 & 0:29 & 3:13 \\
         & LOUD-CH  &   0.64 & 3:47 &  8:13 & 14:51 & 18:38 & 0:13 & 2:31 & 0:29 & 3:13 \\
         & LOUD-CCH &   0.69 & 3:47 &  8:13 & 14:51 & 18:38 & 0:13 & 2:31 & 0:29 & 3:13 \\
  \addlinespace
  Ruhr   & MATSim-h &  62.17 & 5:57 & 12:35 & 13:05 & 19:02 & 1:09 & 3:23 & 0:33 & 5:05 \\
  1pct   & MATSim-e & 201.50 & 5:54 & 12:22 & 13:52 & 19:46 & 1:04 & 3:19 & 0:33 & 4:56 \\
         & LOUD-CH  &   3.52 & 5:54 & 12:22 & 13:52 & 19:46 & 1:04 & 3:19 & 0:33 & 4:56 \\
         & LOUD-CCH &   3.65 & 5:54 & 12:22 & 13:52 & 19:46 & 1:04 & 3:19 & 0:33 & 4:56 \\
  \addlinespace
  Ruhr   & MATSim-h & 108.14 & 3:36 &  8:17 & 13:12 & 16:49 & 0:24 & 2:41 & 0:30 & 3:35 \\
  10pct  & MATSim-e & 223.06 & 3:43 &  8:43 & 13:54 & 17:38 & 0:22 & 2:34 & 0:30 & 3:25 \\
         & LOUD-CH  &   2.44 & 3:43 &  8:43 & 13:54 & 17:38 & 0:22 & 2:34 & 0:30 & 3:25 \\
         & LOUD-CCH &   2.77 & 3:43 &  8:43 & 13:54 & 17:38 & 0:22 & 2:34 & 0:30 & 3:25 \\
  \bottomrule
\end{tabular*}
\end{table*}

We observe that LOUD is 30~times (20~times) faster than MATSim-h on Berlin-10pct (Berlin-1pct). On Ruhr-10pct, we even see a speedup of around 45. Since MATSim-e and both LOUD variants are exact algorithms, all three make the same matching decisions, and thus obtain the same solution quality. Interestingly, although MATSim-h does not find the best insertion for each individual ride request, it obtains slightly better wait times in total on Berlin-10pct.

On Ruhr-10pct, MATSim-h achieves both better wait and ride times than the exact algorithms. Note, however, that the vehicles operate less efficiently. That is, MATSim-h utilizes the fleet under the given constraints worse than the three other algorithms.

\subsection{Integrating LOUD into the MATSim Software Package.}

\begin{table*}[tb]
  \caption{Running time and solution quality of the MATSim mobsim when using the built-in dispatching algorithm, our reimplementations of the built-in algorithm, and our LOUD implementations. We consider various benchmark instances, each without traffic (vehicles travel at free-flow speed) and with traffic (travel speeds depend on how many vehicles are on a segment). For requests, we report the average and 95th percentile of the wait times, and the average ride and trip time. For vehicles, we report the average distance driven empty, occupied, and in total.}
  \label{tab:matsim-mobsim}
  \begin{tabular*}{\linewidth}{
    l@{\extracolsep{\fill}}c@{\extracolsep{\fill}}l@{\extracolsep{\fill}}r@{\extracolsep{\fill}}
    r@{\extracolsep{\fill}}r@{\extracolsep{\fill}}r@{\extracolsep{\fill}}r@{\extracolsep{\fill}}
    r@{\extracolsep{\fill}}r@{\extracolsep{\fill}}r}
  \toprule
  & & & & \multicolumn{4}{c}{\makebox[0pt]{request stats [m:s]}} & \multicolumn{3}{c}{\makebox[0pt]{vehicle stats [km]}} \\
  \cmidrule{5-8}\cmidrule{9-11}
  & & & time & \multicolumn{2}{c}{wait} & ride & trip & emp & occ & total \\
  \cmidrule{5-6}
  input & \hspace{-\tabcolsep}traf\hspace{-\tabcolsep} & algorithm & [h:m] & avg & \hspace{-\tabcolsep}95\,\%ile & & & & & \\
  \midrule
  Berlin &  $\circ$  & built-in &  0:59 & 4:15 &  8:44 & 14:44 & 19:00 & 11.9 & 91.9 & 103.8 \\
  1pct   &           & MATSim-h &  0:07 & 4:16 &  8:35 & 14:32 & 18:48 & 12.2 & 91.8 & 104.0 \\
         &           & MATSim-e &  0:08 & 4:17 &  8:34 & 14:32 & 18:50 & 12.3 & 91.7 & 104.0 \\
         &           & LOUD-CH  &  0:03 & 4:16 &  8:35 & 14:33 & 18:49 & 12.1 & 91.5 & 103.6 \\
         &           & LOUD-CCH &  0:03 & 4:17 &  8:36 & 14:32 & 18:49 & 12.3 & 91.8 & 104.1 \\
  \addlinespace
  Berlin & $\bullet$ & built-in &  0:56 & 5:48 & 14:08 & 19:05 & 24:53 & 13.0 & 92.1 & 105.1 \\
  1pct   &           & MATSim-h &  0:07 & 5:50 & 13:50 & 18:52 & 24:42 & 13.9 & 92.6 & 106.5 \\
         &           & MATSim-e &  0:09 & 5:53 & 14:14 & 18:50 & 24:42 & 14.1 & 92.3 & 106.4 \\
         &           & LOUD-CH  &  0:03 & 5:58 & 14:35 & 18:45 & 24:43 & 13.9 & 92.4 & 106.3 \\
         &           & LOUD-CCH &  0:03 & 5:51 & 14:10 & 18:43 & 24:33 & 13.8 & 92.3 & 106.1 \\
  \addlinespace
  Berlin &  $\circ$  & built-in & 12:41 & 3:53 &  8:53 & 16:30 & 20:24 &  3.8 & 64.8 &  68.7 \\
  10pct  &           & MATSim-h &  1:10 & 3:59 &  8:50 & 15:50 & 19:49 &  4.5 & 66.6 &  71.1 \\
         &           & MATSim-e &  1:25 & 4:01 &  8:42 & 15:49 & 19:50 &  4.3 & 66.5 &  70.8 \\
         &           & LOUD-CH  &  0:29 & 4:02 &  8:46 & 15:50 & 19:52 &  4.4 & 66.5 &  70.9 \\
         &           & LOUD-CCH &  0:29 & 4:02 &  8:47 & 15:46 & 19:48 &  4.4 & 66.5 &  70.9 \\
  \addlinespace
  Berlin & $\bullet$ & built-in & 12:57 & 4:34 & 10:42 & 19:28 & 24:02 &  4.2 & 65.4 &  69.6 \\
  10pct  &           & MATSim-h &  1:13 & 5:01 & 11:15 & 19:34 & 24:35 &  5.1 & 67.6 &  72.8 \\
         &           & MATSim-e &  1:27 & 5:11 & 11:21 & 19:37 & 24:48 &  5.1 & 67.5 &  72.6 \\
         &           & LOUD-CH  &  0:32 & 5:18 & 11:35 & 19:43 & 25:01 &  5.1 & 67.7 &  72.8 \\
         &           & LOUD-CCH &  0:32 & 5:18 & 11:30 & 19:37 & 24:55 &  5.2 & 67.5 &  72.7 \\
  \bottomrule
\end{tabular*}
\end{table*}

MATSim is a full-fledged software package for transport simulations that is currently used in industry and academia. It offers support for a wide variety of transportation types, including driving, walking, transit, cycling, and ridesharing systems. In the previous section, we reimplemented the algorithm used by MATSim to dispatch shared taxi-like vehicles and compared it to our LOUD algorithm. In this section, we go the other way around and integrate our LOUD implementation into the MATSim software package.

MATSim simulates the movement of each inhabitant (also called \emph{agent}) in the study area. For each agent, MATSim maintains a set of alternative \emph{day plans}, which consist of a sequence of activities at different locations and trips between these locations. Trips can use driving, walking, transit, cycling, ridesharing systems, and more modes of transportation. For each trip, a day plan contains the full route that the agent will take. Moreover, a \emph{score} is associated with each day plan that represents the plan's fitness or attractiveness.

The goal of MATSim is to predict the movement of a population, i.e., to generate realistic day plans. To do so, MATSim operates in iterations. Each iteration consists of the three phases replanning, mobsim (for \underline{mob}ility \underline{sim}ulation), and scoring. At the beginning of each iteration, each agent selects one of its day plans based on their scores. A certain fraction of the agents is allowed to modify their selected day plan, for example by changing the time or location of an activity or the route or mode of a trip. During the mobsim, the agents move along the routes determined by their day plan. At the end of each iteration, each agent associates a score with their day plan based on how well the plan worked. MATSim stops when the scores have converged. The final day plans then realistically predict the movement of the population in the study area.

We are mainly interested in the mobsim, which moves the agents along their routes as follows. The mobsim associates a FIFO queue with each edge in the road network. Moreover, each edge has a free-flow travel time, a flow capacity, and a storage capacity. During the mobsim, the agents move from queue to queue along their routes. In each time step, the mobsim repeatedly moves the agent at the head of each queue to the tail of the next queue as long as three conditions hold. First, the agent has spent at least the free-flow travel time in the queue. Second, the number of agents removed from the queue in the current time step is below the flow capacity. Third, the number of agents in the next queue is below the storage capacity.

In addition to moving agents from queue to queue along preplanned routes, the mobsim also processes dynamic ride requests. For each request in a time step, the mobsim computes the best insertion and modifies the planned route of the selected vehicle accordingly. The shared vehicles themselves move through the queues along their route as any other agent.

MATSim is written in the Java programming language. Its functionality is organized into modules called \emph{contributions}. Support for ridesharing systems is provided by the \texttt{drt} contribution. At its heart is a component called \texttt{DefaultDrtOptimizer}, which is invoked whenever a ride request is received. It picks a vehicle and modifies its route accordingly or rejects the request. To make LOUD accessible within MATSim, we implement a \texttt{LoudDrtOptimizer} that resembles the functionality of the \texttt{DefaultDrtOptimizer}, but uses LOUD rather than the built-in dispatching algorithm.

We do not reimplement LOUD in Java. Instead, the \texttt{LoudDrtOptimizer} accesses our native LOUD implementation through the Java Native Interface (JNI). The native code takes the pickup and dropoff spot and the earliest departure time and returns the best insertion, including the full paths to the pickup, from the pickup, to the dropoff, and from the dropoff. Note that the vehicle routes are maintained twice. The native code needs the routes for the matching decisions and maintains them as discussed in \cref{sec:maintaining-feasibility}. The Java code needs the routes to move the shared vehicles through the queues. We reuse the existing Java code for maintaining and modifying the routes in Java.

\Cref{tab:matsim-mobsim} compares the running time and solution quality of the mobsim when using the built-in dispatching algorithm, our reimplementations of the built-in algorithm, and our LOUD implementations. We consider each instance without and with traffic. Note that in this experiment, we make matching decisions based on free-flow travel times. That is, all computed insertions and paths are optimal with respect to the free-flow metric. Without traffic ($\circ$), all agents move at free-flow speed through the queues. We achieve this by multiplying the flow capacities and storage capacities by 100, which essentially makes traffic congestion impossible. With traffic ($\bullet$), the agents may get stuck in traffic jams, leading to delays in the arrival of some agents.

MATSim reports the running time for each phase (replanning, mobsim, scoring) of an iteration separately. However, it does not further subdivide the running time of the mobsim. Therefore, the running time reported in \cref{tab:matsim-mobsim} includes not only the time for the matching decisions but also the effort to maintain the vehicle routes in Java and the effort to move the agents of \emph{all} transportation types through the queues. That is, we cannot expect to see the speedups reported in the previous section in this experiment.

On Berlin-1pct, we decrease the total running time of the mobsim from around one hour to only three minutes. Given that we replaced only the dispatcher for the ridesharing system (and reuse any other code as is), this is a considerable improvement. Likewise, the running time on Berlin-10pct is significantly reduced from 13~hours to half an hour, again by replacing only the dispatcher. When using our LOUD implementation within MATSim, the matching decisions are no longer the performance bottleneck of the mobsim.

We observe that the mobsim variants based on MATSim-e, LOUD-CH and LOUD-CCH obtain slightly different request and vehicle statistics. This may be surprising because all three are exact algorithms and thus should make the same matching decisions. The reason for the divergence is that shortest paths generally are not unique. Even when MATSim-e, LOUD-CH and LOUD-CCH obtain the same insertion, they can still return different paths to and from the pickup and dropoff spot. The actual paths, however, can affect all subsequent ride requests, since the detour to service a request depends on the current locations of the vehicles, which in turn depend on the paths of the vehicles.

We do not run into such trouble in our discrete-event simulation, since we essentially teleport the vehicles from stop to stop rather than moving them along their routes. When the location of a vehicle currently driving from stop~$s$ to stop~$s'$ is needed, we run a CH search from $s$ to $s'$, retrieve the actual path, and traverse the path (starting at the departure time at $s$) until we reach the current point in time. Since we use the same method to retrieve the current location for all three algorithms, the matching decisions do not diverge.

\section{Conclusion and Future Work.}
\label{sec:conclusion}

We presented LOUD, a novel algorithm for large-scale dynamic ridesharing. Unlike most competitors, we do not require a huge number of calls to Dijkstra's algorithm, but adapt a modern route planning technique developed for the many-to-many problem (bucket-based contraction hierarchies). Our experiments on the state-of-the-art Open Berlin Scenario with \num{10000}~vehicles and more than \num{100000}~ride requests show that LOUD answers a request in less than a millisecond, which is 30~times faster than current algorithms. This gives plenty of leeway for interactive applications on cities even larger than Berlin. For transport simulations, LOUD is even more important. Since simulators process each request hundreds of times, running time is an even bigger issue than in interactive applications, and requests cannot be answered ``fast enough''.

Since the special-case treatments take up the largest fraction of the running time, it would be interesting to eliminate the two remaining (local) Dijkstra searches. A possible approach would be to maintain \emph{additional} buckets that store the \emph{unpruned} forward CH search spaces of the ends of the current vehicle routes. Note that we cannot apply elliptic pruning because the leeway is unbounded. Instead, we can keep the buckets sorted (e.g., using search trees), which allows us to stop a bucket scan when we visit an entry that cannot possibly yield an insertion better than the currently best one.

Parallelization could also be a key to better performance. Most likely, this would be a combination of fine-grained parallelism and parallelization over several requests. Independent of the internals of LOUD, the main issue here is that a change caused by an earlier request can affect all subsequent requests. Therefore, it would be interesting to investigate how independent requests can be identified or alternatively how dependencies can be detected and repaired. One could also study to what extent certain dependencies can be ignored without severely affecting solution quality.

Finally, it would be interesting to increase the solution space. For example, one could allow requests already matched to a vehicle to be reordered or moved to a different vehicle. Another interesting project are variable pickup and dropoff spots, where riders agree to walk a short distance to a location where it is more efficient to pick them up or drop them off (e.g., on main roads rather than in traffic-calmed areas). We believe that the techniques developed for LOUD might be key ingredients for such generalized systems that promise higher overall solution quality.

This is the full version of a conference paper~\cite{BuchholdSW21} that was presented at the 23rd SIAM Symposium on Algorithm Engineering and Experiments (ALENEX'21).

\bibliography{References}

\end{document}